\documentclass[10pt,twocolumn,twoside]{IEEEtran}

\usepackage[active]{srcltx} 
\usepackage[cmex10]{amsmath}
\usepackage{amsfonts}
\usepackage[final]{graphics}
\usepackage[final]{graphicx}
\usepackage{amsbsy}
\usepackage{amsbsy}
\usepackage{amssymb}
\usepackage{url}
\usepackage{enumerate}
\usepackage{cite}
\usepackage{psfrag}
\usepackage{color}
\usepackage{euscript}
\usepackage[utf8]{inputenc}
\usepackage{algorithmic}
\usepackage[ruled,vlined]{algorithm2e}
\usepackage{soul}
\usepackage{tikz}
\usepackage{textcomp}
\usepackage{lipsum}
\usepackage[T1]{fontenc}
\usepackage{booktabs}
\usepackage{multirow}
\usepackage{siunitx}

\newcommand{\bm}[1]{{\mathbf{#1}}}
\newcommand{\Es}{{\mathbb{E}}}          
\newcommand{\rank}{{\text{rank}}}
\newcommand{\diag}{{\text{diag}}}
\newcommand{\trace}{{\text{tr}}}

\newcommand{\I}{\bm{I}}
\newcommand{\Zero}{\bm{O}}

\newcommand{\Lcp}{L_{\text{cp}}}

\newcommand{\z}{\bm{\widetilde{y}}}
\newcommand{\xa}{x_{\text{U}}}
\newcommand{\xak}{x_{\text{U},k}}
\newcommand{\ba}{s_{\text{U}}}
\newcommand{\bs}{s_{\text{J}}}
\newcommand{\yb}{\bm y}
\newcommand{\Eb}{\bm E}
\newcommand{\Fb}{\bm F}
\newcommand{\Pb}{\bm P}

\newcommand{\alphalabel}{\boldsymbol{\alpha}_{\overline{\yb}\overline{\yb}^*}}
\newcommand{\alphalabelhat}{\widehat{\boldsymbol{\alpha}}_{\overline{\yb}\overline{\yb}^*}}
\newcommand{\Pblabel}{\bm{P}_{\overline{\yb}\overline{\yb}^*}}
\newcommand{\wb}{\bm w}
\newcommand{\Ab}{\bm A}
\newcommand{\Rb}{\bm R}
\newcommand{\vb}{\bm v}
\newcommand{\Bb}{\bm B}

\newcommand{\bab}{\bm s_{\text{U}}}
\newcommand{\bsb}{\bm s_{\text{J}}}
\newcommand{\Hatildezero}{\overline{\bm H}_{\text{U},0}}
\newcommand{\Hatildeuno}{\overline{\bm H}_{\text{U},1}}
\newcommand{\Hstildezero}{\overline{\bm H}_{\text{J},0}}
\newcommand{\Hstildeuno}{\overline{\bm H}_{\text{J},1}}
\newcommand{\Ha}{\bm H_{\text{U}}}
\newcommand{\Hs}{\bm H_{\text{J}}}
\newcommand{\Htilde}{\bm{\widetilde{H}}}
\newcommand{\Hover}{\bm{{G}}}
\newcommand{\htilde}{\bm{\widetilde{h}}}

\newcommand{\wtilde}{\bm{\widetilde{w}}}

\newcommand{\xs}{x_{\text{J}}}
\newcommand{\xsk}{x_{\text{J},k}}
\newcommand{\Ka}{K_{\text{U}}}
\newcommand{\Ks}{K_{\text{J}}}
\newcommand{\Cset}{\mathbb{C}}
\newcommand{\Rset}{\mathbb{R}}

\newcommand{\Zset}{\mathbb{Z}}

\newcommand{\eqdef}{\triangleq}

\newcommand{\herm}{\text{H}}
\newcommand{\trasp}{\text{T}}

\newcommand{\pot}{\EuScript{P}}
\newcommand{\minitab}[2][l]{\begin{tabular}#1 #2\end{tabular}}

\def\bdm#1\edm{\begin{displaymath}#1\end{displaymath}}
\def\be#1\ee{\begin{equation}#1\end{equation}}
\def\barr#1\earr{\begin{align}#1\end{align}}

\newcommand{\IeeeTIT}{{\em IEEE Trans.\ Inf. Theory\/}}
\newcommand{\IeeeTSP}{{\em IEEE Trans.\ Signal Process.\/}}

\newcommand{\IeeeTWC}{{\em IEEE Trans.\ Wireless Commun.\/}}
\newcommand{\IeeeJSAC}{{\em IEEE J.\ Select.\ Areas Commun.\/}}
\newcommand{\IeeeTVT}{{\em IEEE Trans.\ Veh. Technol.\/}}

\newcommand{\IeeeCOMMMAG}{{\em IEEE Commun.\ Magazine\/}}

\newcommand{\IeeeSENSJ}{{\em IEEE Sens. J.\/}}

\newtheorem{theorem}{Theorem}[section]

\newtheorem{proposition}[theorem]{Proposition}

\newcommand\acceptedtext{%
  \footnotesize This article has been accepted for publication in a future issue of this journal, but has not been fully edited. Content may change prior to final publication. \\
  Citation information: DOI 10.1109/JIOT.2021.3132381, IEEE Internet of Things Journal.}
\newcommand\acceptednotice{%
\begin{tikzpicture}[remember picture,overlay]
\node[anchor=north,yshift=-6pt] at (current page.north) {%
\begin{minipage}{\textwidth}
\center \acceptedtext
\end{minipage}};
\end{tikzpicture}%
}

\begin{document}
\title{Detection and blind channel estimation for UAV-aided wireless sensor networks
in smart \\ cities under mobile jamming attack
}

\author{Donatella~Darsena,~\IEEEmembership{Senior Member,~IEEE},
             Giacinto~Gelli,~\IEEEmembership{Senior Member,~IEEE},
        Ivan~Iudice, \\ and Francesco~Verde,~\IEEEmembership{Senior Member,~IEEE}
\thanks{
Manuscript received May 24, 2021; revised October 26, 2021;
accepted November 29, 2021. Date of publication xx yy, 2021; date of current version November 30, 2021. The associate editor coordinating the review of this manuscript and approving it for publication was
Prof.\  Muhammad Khurram Khan \textit{(Corresponding author: Francesco Verde)}.}
\thanks{
D.~Darsena is with the Department of Engineering,
Parthenope University, Naples I-80143, Italy (e-mail: darsena@uniparthenope.it).
G.~Gelli and F.~Verde are with the Department of Electrical Engineering and
Information Technology, University Federico II, Naples I-80125,
Italy [e-mail: (gelli,f.verde)@unina.it].
I.~Iudice is with Italian Aerospace Research Centre (CIRA),
Capua I-81043, Italy (e-mail: i.iudice@cira.it).}
\thanks{D.~Darsena, G.~Gelli, and F.~Verde are also with
National Inter-University Consortium for Telecommunications (CNIT).}}

\markboth{IEEE INTERNET OF THINGS JOURNAL,~Vol.~xx, No.~yy,~Month~2021}{Anti-jamming communications
for UAV-aided wireless sensor networks in smart cities}

\IEEEpubid{\begin{minipage}{\textwidth}\ \\
\center 
2327--4662~\copyright~2021~IEEE. 
\\ Personal use of this material is permitted. However, permission to use this material for any other purposes \\ must be obtained from the IEEE by sending a request to pubs-permissions@ieee.org.
\end{minipage}}

\maketitle
\acceptednotice

\begin{abstract}
Unmanned aerial vehicles (UAVs) can be integrated into
wireless sensor networks (WSNs) 
for smart city applications in several ways.
Among  them, a UAV can be employed as
a relay in a ``store-carry and forward'' fashion by uploading
data from ground sensors and metering devices and,
then, downloading it to a central unit.
However, both the uploading and downloading phases
can be prone to potential threats and attacks. As a legacy
from traditional wireless networks, the jamming
attack is still one of the major and serious threats to
UAV-aided communications, especially when also 
the jammer is mobile, e.g., it is mounted on a UAV 
or inside a terrestrial vehicle.
In this paper, we investigate
anti-jamming communications
for UAV-aided WSNs operating over
doubly-selective channels 
in the downloading phase.
In such a scenario, the signals transmitted by
the UAV and the malicious mobile jammer
undergo both time dispersion
due to multipath propagation effects
and frequency dispersion 
caused by their mobility.
To suppress high-power jamming signals,
we propose a blind physical-layer technique
that jointly detects the UAV and jammer symbols
through serial disturbance cancellation
based on symbol-level post-sorting of the
detector output.
Amplitudes, phases, time delays, and
Doppler shifts  -- required to implement 
the proposed detection
strategy -- are blindly estimated from data
through the use of algorithms that exploit
the almost-cyclostationarity properties of the
received signal and
the detailed
structure of multicarrier 
modulation format. Simulation results corroborate
the anti-jamming capabilities
of the proposed method, for different
mobility scenarios of the jammer.
\end{abstract}


\section{Introduction}

\IEEEPARstart{W}{ireless} sensor networks (WSNs)
are expected to play a fundamental role
in realizing the
vision of future {\em smart cities} \cite{Ham.2019, Meh.2017,Hase.2021,Zaman.2021},
which represent one of the major application of
Internet-of-Things (IoT).
In a smart city, a huge number
of sensors and metering devces of various nature
are placed at a large variety of
locations, such as  buildings, parking areas,
traffic lights, railway/metro/bus stations,
monitoring hubs along roads, and so on,
whose task is to gather data to
be transmitted towards a powerful
central unit (CU) located
somewhere in the city.
Processing of the collected data
at the CU enables provision of
many innovative services within health,
transportation, sustainability, economy,
agriculture, law  enforcement,  community, and
other areas affecting the overall wellbeing
of the residents and businesses.
Recently, it has been pointed out the crucial
role of smart cities in reducing
the COVID-19  risk \cite{Jah.2020}.

Despite the promising service features,
the problem of efficient and reliable
transmission to a central
unit
of the data independently collected
by sensors and meters
remains a challenging issue in IoT applications \cite{Kul.2017}.
One possible option is
represented by the LoRa/LoRaWAN
technology stack \cite{Oso.2020},
due to its long signal range and
minimal power requirements. However, it
is common belief that, due to their limited
capacity, these technologies
could be unable to support
the high data-rates required
by future smart city applications.
In this respect, the employment
of {\em unmanned aerial vehicles (UAVs)} as aerial
``store-carry and forward'' relay
nodes for distributed ground sensors
and meters is expected to
bring significant benefits to
future WSNs, such as higher data transmissions,
larger service coverage, and reduced energy consumption
\cite{Hayat.2016, Zhang.2020}.
According to this paradigm, UAVs can move
between a sensor/meter and the CU, by gathering
data from the sources ({\em uploading phase}),
and, then, transmitting them to
the destination ({\em downloading phase}).

Unfortunately, while the enabling technologies for UAV-aided
WSNs gradually mature, all kinds of potential threats or attacks will also
rise  and endanger smart city applications \cite{Wu-Mei.2019}.
These adversarial behaviors are
due to malicious users, such as
criminals, terrorists, and business spies,
and are driven by various motivations, such as
committing crimes, jeopardizing public safety,
invading secret databases, and so on.

\IEEEpubidadjcol

\subsection{Jamming attacks and related countermeasures}

One serious threat to UAV-aided WSNs
is the {\em jamming attack}.
Radio jamming can be divided
into \textit{link-layer} and \textit{physical-layer jamming}.
The aim of link-layer jamming is
to decrease the packet
send/delivery ratio \cite{Zou.2016,Pele.2011}.
However, these attacks are based
on the assumption that the jammer
is aware of the link-layer
protocol of the legitimate user,
which could be an excessively pessimistic assumption.
On the other hand, physical-layer
jamming methods are very
effective in UAV-aided WSNs, since
a jammer might emit high-power electromagnetic
signals to make the
legitimate signals unrecognizable
for the UAV in the uploading phase
and/or for the CU in the downloading one.
Moreover, due to the rapid development
of software-defined radio techniques,
a smart jammer is able to modify
the attack pattern according to the
transmission features
of the targeted communication links.

Because of the ease of implementation and
disruptive impact, physical-layer jamming
attacks in UAV-aided WSNs have
received a lot of attention
and the countermeasures
have been extensively studied.
Some works rely on strategies implemented at
the application layer to protect networks from external attacks.
A viable indicator to verify the existence of jamming attacks in WSNs
is represented by the packet loss ratio \cite{Zarpel}.
An intelligent security scheme has been proposed in \cite{Saeed},
which uses random neural networks
to detect malicious activities in the network and 
verify the authenticity and legitimacy of the network traffic.
A lightweight authentication scheme has been proposed \cite{Wazid-1}
in a cloud-based IoT infrastructure to ensure the
security of a collected data from a sensor at a remote location by means of a one-way cryptographic
hash function authentication process.
In \cite{Wazid-2}, a one-way cryptographic hash function with
addition of bitwise exclusive-OR (XOR) operation to verify the legitimacy of participating devices in the network
has been developed.
However, applying such application-layer countermeasures
to UAV-aided WSNs is very difficult,
due to their limited resources and dynamic nature. Moreover,
they are mainly targeted at identifying
the presence of a jamming attack and, thus, are not able 
to avoid and/or combat it.
Furthermore, most of them 
are specific to a given environment, system, or software, 
and exhibit a complex
authentication process, which generates network
overhead.
Therefore, innovative anti-jamming strategies are needed at
the physical layer to counteract jamming threats.

Anti-jamming physical-layer methods can be categorized
as \textit{proactive} or \textit{reactive}
ones.
Proactive strategies
either hide signals so that
the jammer cannot detect
them (like, e.g., in direct sequence spread spectrum)
or avoid the jamming attack by
periodically changing channels with
an interval of fixed length, without checking
for jammers \cite{Pele.2011} (like, e.g., in
frequency-hopping spread spectrum).
On the other hand, reactive strategies
suppress the jamming signal by signal processing
techniques or combat the jamming attack
directly, e.g., by increasing the transmission power.

Proactive approaches are not particularly suited
to UAV-aided WSNs,
due to their scarce
adaptability to dynamic spectrum state
and jamming patterns, as well as for
their low spectral efficiency \cite{Wang.2020}.
As to reactive approaches
for conventional wireless networks,
they can be classified into
two broad families.
Techniques belonging to the first
family aim at optimizing
the transmission power
for the legitimate user and/or at
developing efficient channel-selection strategies
after detection
of channel jamming.
In \cite{Yang.2013},
the power-control problem for the legitimate
user was studied from a Stackelberg game perspective.
By utilizing
the spectrum waterfall representation, an anti-jamming
scheme based on deep reinforcement learning (RL) method
was proposed in \cite{Liu.2018} to facilitate the channel-selection process.
A multi-armed bandit framework was formulated in
\cite{Zhou.2016} to obtain efficient channel-selection strategies.
In \cite{Jia.2018}, a multi-domain anti-jamming scheme that tackles
both power control and channel selection was proposed for
heterogeneous wireless networks.

Reputation-based anti-jamming methods belong to the second family.
In \cite{Xiao.2012},  a reputation updating process
was developed for large-scale wireless networks in order to face with
a wide range of attacks. The problem of establishing  the  defender's reputation
in anomaly detection  against insider attacks was  addressed  in \cite{Zhang.2010}.
However, due to different network characteristics, such methods are not directly
applicable to UAV-aided WSNs.

Specific countermeasures against jamming attacks in UAV-aided
communications have been proposed in many recent papers.
In \cite{Xu.2018}, power control for anti-jamming problem
in a pure UAV communication network was investigated from
a game theoretic perspective.
Multi-parameter programming and RL was used in \cite{Peng.2019}
to protect against
multiple jammers the communication between a UAV swarm 
and a base station.
A UAV-aided cellular
framework was studied in \cite{Lu.2020}, in which a UAV uses
RL to choose the relay policy
for a mobile user whose serving base station is
attacked by a jammer.
In \cite{Xiao.2018}, anti-jamming strategies based on
RL were investigated to improve
UAV-enabled communications against smart jamming attacks.
By taking advantage of the flexible mobility patterns
of UAVs, anti-jamming
trajectories were designed in \cite{Wu.2019,Mah.2019,Duo.2021}
to prevent ground jammers  from degrading the legitimate transmission
in UAV-enabled communication systems.

Many of the above works assume a
non-dispersive channel
model for ground-to-air (G2A)
and air-to-ground (A2G) links.
However, such an ideal assumption
is inaccurate in smart city applications,
where  UAVs are expected to be deployed
in urban areas with high and dense
buildings or obstacles. In these scenarios,
since UAVs are highly mobile in nature,
the G2A/A2G channels are
better modeled as \textit{doubly-selective}
\cite{Tse-book} channels, with multipath effects causing time
dispersion and Doppler shifts giving rise to frequency dispersion of the
information-bearing signals.
Additionally, the jammer can be {\em mobile} in its turn, e.g., jamming
radios can be piggybacked on a UAV or installed
inside terrestrial vehicles. 
Mobility allows the jammer to rapidly move
around the city and complicates localization of the jamming source
compared to the case of a fixed jammer.

Henceforth, mobile jamming attacks
specifically targeted at UAV-aided WSNs remain
a critical issue, and
effective anti-jamming reactive methods operating over
doubly-selective channels are an open research issue.

\begin{table*}
\centering
\caption{List of the main notations used throughout the paper.}
\label{tab:notation}
\begin{tabular}{ll} 
\bf{Notation} & \bf{Meaning}\\
 \hline
 \hline
 $x$ & scalar value \\
$\bm A$ & matrix -- bold, capital letter \\
$\bm a$ & vector -- bold, small letter \\
$x^*$ & conjugate of $x$ \\
$\bm A^\trasp$ & transpose of $\bm A$ \\
$\bm A^\herm$ & Hermitian (conjugate transpose) of $\bm A$ \\
$\bm A^{-1}$ & inverse of $\bm A$ \\
$\mathbb{C}$ & fields of complex numbers \\
$\mathbb{R}$ & fields of real numbers \\
$\mathbb{Z}$ & fields of integer numbers \\
$\mathbb{C}^{n}$  &  vector-space of all $n$-column vectors with complex coordinates \\
$\mathbb{R}^{n}$ &  vector-space of all $n$-column vectors with real coordinates \\
$\mathbb{C}^{n \times m}$ & vector-space of all the $n \times m$ matrices with complex elements \\
$\mathbb{R}^{n \times m}$ & vector-space of all the $n \times m$ matrices with real elements \\
$\delta_n$ & Kronecker delta \\
$j \eqdef \sqrt{-1}$ & denotes the imaginary unit
\\
$\text{max}(x,y)$ & maximum between $x \in \mathbb{R}$ and $y \in \mathbb{R}$\\
$\lfloor x \rfloor$ & rounds $x \in \mathbb{R}$ to the nearest integer less than or equal to $x$\\
mod $1$ & denotes modulo $1$ operation with values in $[-1/2,1/2)$ \\
$\mathcal{O}$ & Landau symbol \\
$\ast$ & linear convolution operator \\
$\mathcal{A}^n= \mathcal{A} \times \cdots \times \mathcal{A}$ &
$n$th Cartesian power of the set $\mathcal{A}$ \\
$\mathbf{0}_{n}$ & $n$-column zero vector \\
$\Zero_{n \times m}$ & $n \times m$ zero matrix \\
$\I_{n}$ & $n \times n$ identity matrix \\
$\bm{W}_n$ & unitary symmetric $n$-point inverse discrete Fourier transform (IDFT) matrix \\
$\bm{W}_n^{-1}=\bm{W}_n^\herm$ & $n$-point
discrete Fourier transform (DFT) matrix \\
$\{\mathbf{a}\}_{\ell}$ & $\ell$th entry of $\bm{a} \in \Cset^n$ \\
$\{\mathbf{A}\}_{\ell,\ell}$ & $\ell$th diagonal entry \\
$\lambda_\text{max}(\mathbf{A})$ & largest eigenvalue of $\mathbf{A} \in \Cset^{n \times n}$ \\
$\mathrm{trace}(\mathbf{A})$ & trace of $\mathbf{A} \in \Cset^{n \times n}$ \\
$\mathrm{rank}(\mathbf{A})$ & rank of $\mathbf{A} \in \Cset^{n \times m}$ \\
$\|\Ab\| \eqdef [\trace(\Ab\,\Ab^\herm)]^{1/2}$ & (induced)
Frobenius norm of $\Ab \in \Cset^{n \times m}$ \cite{Horn} \\
$\mathbf{A}= \diag (a_{0}, a_{1}, \ldots, a_{n-1})$ & $n \times n$ diagonal matrix \\
$\text{vec}(\mathbf{A})$ & concatenation of the columns of $\Ab \in \Cset^{n \times m}$ \\
$\langle \cdot \rangle$ & infinite-time temporal averaging \\
$\Es[\cdot]$ & ensemble averaging \\
$x(t)$ & continuous-time signals \\
$x[n]$ & discrete-time signals \\
\hline
\hline
\end{tabular}
\end{table*}

\subsection{Anti-jamming communications in smart city applications}

In this paper, we focus on anti-jamming
communications by explicitly taking into account
the \textit{dispersive} nature
of the wireless channel in smart city applications.
We assume that the mobile jammer itself is {\em smart},
in the sense that it is able to replicate
the same modulation format of
the legitimate transmission \cite{Pele.2011}.
Since multicarrier modulation is currently
adopted in many wireless standards,
we focus on the \textit{orthogonal frequency-division multiplexing} (OFDM)
transmission format.
In this case,
besides the adverse effects of the jamming signal,
due to the UAV speed and fast fading channel conditions,
the performance of the multicarrier communication system is also deteriorated
by Doppler frequency shifts.

The proposed anti-jamming solution
relies on a known information-theoretic result \cite{Carl.1978},
according to which, when the signal-to-jammer
ratio (SJR) is sufficiently
low, the optimal detection strategy
of the legitimate transmission (in the minimum-bit-error-rate sense)
consists of first estimating the jammer signal and, then,
subtracting it from
the received data.
What is  challenging for the problem
at hand is that, in order to implement such
a jamming-resistant receiver,
\textit{all} the amplitudes, phases,
time delays, and Doppler shifts
of both the legitimate and jamming
channels have to be accurately
estimated at the receiver.

Traditional training-based approaches
for parameter estimation (see, e.g., \cite{Fle.1999})
could not work well in
UAV-aided WSNs under a jamming attack, since communication resources
may be significantly wasted in
doubly-selective channels and,
most importantly, the
jamming attack during the training phase
might induce an unacceptable
estimation error.
Motivated by these facts,
to estimate all the parameters of interest of
the legitimate  and jamming signals,
we propose a {\em blind}
estimation algorithm that relies on the
{\em almost-cyclostationarity (ACS) properties}
\cite{Gardner-book} of the received signal.

The proposed algorithm capitalizes on two
distinct features of UAV-aided WSNs
for smart city applications:
(i) the availability of a large amount of data,
as a consequence of the  widespread
dissemination of sensors and meters;
(ii) the relaxed latency constraints \cite{Li.2015},
due to the fact that the network scope
is not to react to an instantaneous
event but rather to monitor the history
of sensed data, based on which it is possible to provide new services or
to optimize resource consumption.
To the best of our knowledge, this is the first paper
that investigates the application of
ACS properties
to achieve  anti-jamming communications in
UAV-aided WSNs.

\subsection{Contribution and organization}

The contributions of the paper
are summarized as follows:

\begin{enumerate}[(i)]

\item
As a radical alternative to  existing methods,
we propose an anti-jamming detector
leveraging on the
structure of the jamming
signal rather than treating
it as random noise.
This allows one to greatly reduce
the negative impact of mobile jamming.
At a high-level, such an idea is akin to
multiuser/multiantenna
detection \cite{Ser.2004},
but at a low-level it is significantly different,
since the mobile jammer is not a constituent part
of the network and, thus, its transmitter
cannot be designed to ease
separation of the UAV and jamming signals
at the receiver.

\item
We propose blind algorithms based
on ACS properties to
estimate the channel
parameters of both the UAV and
jamming transmissions,
by developing a simple rule unveiling
the association between
the estimated parameters
and the two superimposed signals.
The developed estimators extend and generalize
previously-proposed ACS-based
estimation techniques
\cite{Chen.I, Chen.II,Gelli, Ciblat, Napolitano.2019}
to doubly-selective channels.

\item
The performances of the proposed channel
estimation and detection methods are validated using
customary metrics.
It is demonstrated that the combination of
jamming cancellation for the detection and
ACS-based estimation of channel
parameters confers robustness
against a mobile jamming attack in UAV-based WSNs,
under different mobility scenarios of the jammer.

\end{enumerate}

The paper is organized as follows
(the main notations are reported in Table~\ref{tab:notation}).
A reference framework and a model of the considered
anti-jamming communication system is described
in Section~\ref{sec:system}.
The proposed ACS-based estimators
of the relevant channel parameters are developed
in Section~\ref{sec:estimation}.
Numerical results are reported
in Section~\ref{sec:simul}.
Finally, the main results obtained
in the paper are summarized in Section~\ref{sec:concl}.

\begin{figure}[t]
\centering
\includegraphics[width=\columnwidth]{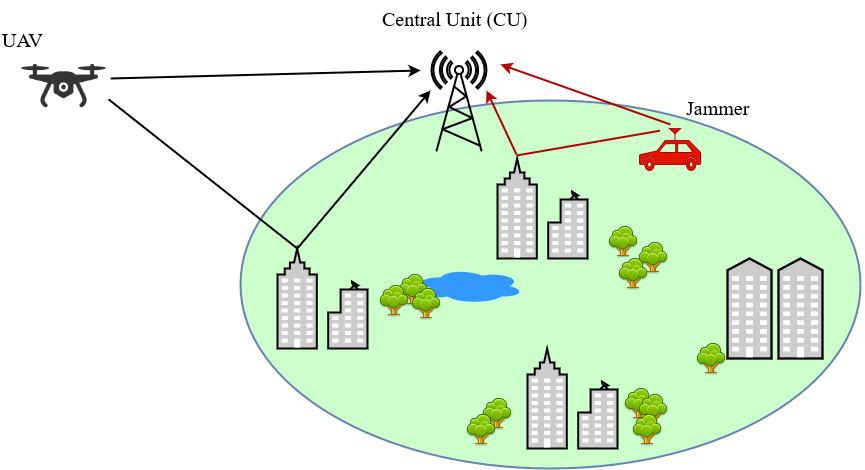}
\caption{UAV-aided WSN for smart city 
applications (downloading phase) in a multipath scenario: 
the UAV transmits data gathered from ground
sensor/meter devices
to a CU in the presence of a smart mobile jammer,
whose aim is to degrade the UAV-to-CU transmission.}
\label{fig:fig_1}
\end{figure}

\section{The reference framework}
\label{sec:system}

With reference to Fig.~\ref{fig:fig_1}, we focus on the downloading phase
where the UAV transmits to the CU, over a dedicated A2G link,
the data collected by the sensors
and meters disseminated around the city
(the proposed anti-jamming strategy can be used in
the uploading phase with minor modifications).
An hostile mobile terminal, referred to as the jammer, is trying to
disrupt the communication
link between the UAV and the CU, by sending  jamming signals.
To cope with the time-dispersive
nature of the channel,
the UAV employs OFDM transmission
with OFDM symbol period $T$ and
sampling period $T_\text{c}$.

The UAV transmits data to the CU in a mobility-based fashion
aimed at improving the physical layer security \cite{Tang} of the link.
Specifically, we will show that nonzero radial velocities between
the UAV and the CU allow to gain robustness against
a mobile jamming attack by exploiting their different
features in the Doppler domain.

Also the jammer is mobile, e.g., 
the jamming transmitter can be
mounted on a UAV or inside a terrestrial vehicle (as depicted in
Fig.~\ref{fig:fig_1}). 
From its viewpoint, mobility allows the jammer to achieve two
desirable goals: 
(i)  to move  quickly from one point to another 
of the city in order to jam multiple CUs;
(ii) to  drastically reduce 
its probability of localization.

\subsection{Basic assumptions}
\label{sec:ass}
In this subsection, we state and discuss all the assumptions that are
subsequently used throughout the paper.

\vspace{2mm}
{\bf A1)}:  The jammer is able to replicate almost
perfectly the modulation format
of the UAV, by injecting in the air
random symbols that are
independent of those transmitted
by the UAV. Compared to the simpler case when
the jammer just sends a high-power unmodulated carrier to
corrupt the symbol stream of the UAV transmission,
the considered attack is harder to be detected using network monitoring tools,
since such tools will sense legitimate traffic on the medium \cite{Pele.2011}.

\vspace{2mm}
{\bf A2)}:
The independent symbol streams $\{\ba[n]\}$ and $\{\bs[n]\}$
emitted  by the UAV and the mobile jammer ($n \in \mathbb{Z}$), respectively, are
modeled as sequences of
zero-mean unit-variance
independent and identically distributed (i.i.d.) complex
noncircular random variables,
with second-order moments $\Es(\ba^2[n]) \neq 0$ and $\Es(\bs^2[n]) \neq 0$,
respectively.%
\footnote{Throughout the
paper, the subscripts U and J
indicate parameters referring to the UAV and the jammer, respectively.}
Real-valued modulation schemes,
such as pulse amplitude modulation (PAM) and binary
phase-shift keying (BPSK),
naturally exhibit such noncircular statistical features.
Furthermore, modulations with quadrature offset, e.g., offset
quadrature phase-shift keying (OQPSK), minimum shift keying
(MSK) and variants, are noncircular, too: they are attractive
for systems using nonlinear amplifiers since
the modulated signal is less sensitive to distortion during symbol transitions.
Extension of the proposed anti-jamming strategy to circular modulation schemes
(i.e., when $\Es(\ba^2[n])=\Es(\bs^2[n])=0$) is outlined in Section~\ref{sec:detection}
(see Remark~1) and Section~\ref{sec:estimation} (see Remark~2).

\vspace{2mm}
{\bf A3)}:
The signals transmitted by the
UAV and the mobile jammer propagate through
doubly-selective channels \cite{Bello}, with
maximum delay spread $\Delta_\text{TX}$
(its reciprocal is proportional to the coherence bandwidth)
and  Doppler spread $D_\text{TX}$
(its reciprocal is proportional
to the coherence time) \cite{Proakis},
for $\text{TX} \in \{\text{U}, \text{J}\}$.
We assume underspread channels, i.e., $\Delta_\text{U} \, D_\text{U} < 1$ and
$\Delta_\text{J} \, D_\text{J} < 1$, which is a valid assumption
for most radio channels \cite{Matz}.

\vspace{2mm}
{\bf A4)}:
The CU
receiver has been previously locked to
the multipath component
at (approximately) the minimum delay
$\tau_\text{min}$
between the UAV-to-CU and jammer-to-CU channels,
in which case the parameter
$\Delta_\text{max} \eqdef \max(\Delta_\text{U}, \Delta_\text{J})$ characterizes the time spreading of both channels.

\vspace{2mm}
{\bf A5)}:
The carrier frequency is much larger
than the signaling bandwidth. In this case,
the shrinking or dilation
of the signaling waveform \cite[Ch.~7]{Napolitano.2012}
can be ignored, and
the effects of motion are accurately
captured by Doppler shifts.

\vspace{2mm}
{\bf A6)}:
During the downloading phase,  the UAV flies at
a sufficiently high altitude such that the UAV and the CU have
a line-of-sight (LoS) connection
with probability $1$.

\vspace{2mm}
{\bf A7)}:
For blind identification of all the
channel parameters (see Section~\ref{sec:estimation}),
all the Doppler shifts of the UAV-to-CU
and jammer-to-CU transmissions
are different and all their pairs
have different sums -- a technical condition
that is satisfied with probability
$1$ in real scenarios.

\vspace{2mm}
{\bf A8)}:
If we denote with
$\psi_\text{DAC}(t)$
the impulse response of the
digital-to-analog converter (DAC) and with
$\psi_\text{ADC}(t)$
the impulse response of the (anti-aliasing) filter at the
input of the analog-to-digital converter (ADC) at the CU, the
impulse response  $\psi(t) \eqdef \psi_\text{DAC}(t) \ast \psi_\text{ADC}(t)$
of the cascade of the DAC interpolation filter and the ADC antialiasing
filter obeys $\psi(t) \equiv 0$ for $t \not \in [0, \Delta_\text{filter})$.
The pulse $\psi(t)$
and its finite duration $\Delta_\text{filter}$ are known at the CU and
$\Delta_\text{filter}+\Delta_\text{max} < T$
for multicarrier systems (see also A4).

\vspace{2mm}
{\bf A9)}:
After filtering, the noise observed at the receiver
is a zero-mean complex circular white Gaussian random process
with power $\sigma_w^2$.

\vspace{2mm}
{\bf A10)}:
The OFDM cyclic prefix (CP)  length is sufficient 
to accommodate the overall delay spread, i.e.,
\[
\Lcp \ge (\Delta_\text{filter}+\Delta_\text{max})/T_\text{c} \: .
\]

\begin{table*}[t]
\centering
\caption{List of the main symbols used throughout the paper (with $\text{TX} \in \{\text{U},\text{J}\}$).}
\label{tab:symbol}
\begin{tabular}{lll} 
 & \bf{Symbol} & \bf{Description}\\
 \hline
 \hline
& $M$ & number of subcarriers \\
\cline{2-3}
& $\Lcp$ & cyclic prefix length \\
\cline{2-3}
\multirow{3}*{\minitab[c]{\bf{System} \\ \bf{symbols}}}
& $P$ & OFDM symbol length \\
\cline{2-3}
& $T_{\mathrm{c}}$ & sampling period \\
\cline{2-3}
& $T$ & OFDM symbol period \\
\cline{2-3}
& $f_0$ & carrier frequency \\
\cline{2-3}
& $\pot_{\text{TX}}$ & average RF transmit power of transmitter TX\\
\cline{2-3}
& $\Delta_{\text{filter}}$ & duration of the ADC impulse response \\
\cline{2-3}
& $\bm{s}_\text{TX}[n] $ & block of $M$ data symbols transmitted by TX during the $n$th OFDM symbol interval\\
\hline
\hline
 & $K_{\text{TX}}$ & number of paths for the TX-to-CU channel \\
\cline{2-3}
& $A_{\text{TX},k}$ & amplitude of the $k$th path for the TX-to-CU channel \\
\cline{2-3}
& $\theta_{\text{TX},k}$ & phase shift of the $k$th path for the TX-to-CU channel \\
\cline{2-3}
& $f_{\text{TX},k}$ & Doppler frequency shift of the $k$th path for the TX-to-CU channel \\
\cline{2-3}
\multirow{3}*{\minitab[c]{\bf{Channel} \\ \bf{symbols}}}& $\tau_{\text{TX},k}$ & delay of the $k$th path for the TX-to-CU channel \\
\cline{2-3}
& $\Delta_{\text{max}}$ & maximum channel time spreading \\
\cline{2-3}
& $g_{\text{TX},k}$ & channel gain of the $k$th path for the TX-to-CU channel \\
\cline{2-3}
& $\nu_{\text{TX},k}$ & normalized Doppler shift of the $k$th path for the TX-to-CU channel \\
\cline{2-3}
& $d_{\text{TX},k}$ & integer delay of the $k$th path for the TX-to-CU channel \\
\cline{2-3}
& $\chi_{\text{TX},k}$ & fractional delay of the $k$th path for the TX-to-CU channel \\
\cline{2-3}
& $\bm{H}_{\text{TX}}[n]$ & time-varying $M \times M$ TX-to-CU channel matrix after CP removal \\
\hline
\hline
\end{tabular}
\end{table*}

\subsection{Signal models}
\label{sec:sig}

According to A1,
the flying UAV and the mobile jammer employ $M$ orthogonal  subcarriers.

Let $s_{\text{TX}}^{[m]}[\ell] \eqdef s_{\text{TX}}[\ell \, M + m]$
represent the symbol transmitted in the $\ell$th data block
on the $m$th subcarrier (see A2), with $\text{TX} \in \{\text{U},\text{J}\}$,
the generic symbol vector
$\bm{s}_\text{TX}[n] \eqdef (s_\text{TX}^{[0]}[n], s_\text{TX}^{[1]}[n],\ldots,
s_\text{TX}^{[M-1]}[n])^\trasp \in \Cset^M$ is
subject to IDFT
and CP insertion, thus yielding $\bm{u}_\text{TX}[n] \eqdef (u_\text{TX}^{[0]}[n], u_\text{TX}^{[1]}[n],\ldots,
u_\text{TX}^{[P-1]}[n])^\trasp= \bm{I}_\text{cp} \, \bm{W}_M \, \bm{s}_\text{TX}[n]$,
where $\bm{I}_\text{cp} \in \mathbb{R}^{P \times M}$ models
the insertion of a CP of length $\Lcp$, with $P \eqdef M+\Lcp$ and $\bm{W}_M$ is the $M$-point
IDFT.

The data vector $\bm{u}_\text{TX}[n]$ undergoes parallel-to-serial (P/S) conversion,
and the resulting sequence $\{u_\text{TX}[\ell]\}$
feeds the DAC, operating at rate
$1/T_\text{c}=P/T$, where $u_\text{TX}[\ell \, P+q]=u_\text{TX}^{[q]}[\ell]$,
 for $q \in \{0,1.\ldots, P-1\}$.
The baseband transmitted continuous-time
signal at the output of the DAC is given,
for $\text{TX} \in \{\text{U}, \text{J}\}$, by
\begin{multline}
x_\text{TX}(t) = \sqrt{2 \, \pot_{\text{TX}}} \sum_{\ell=-\infty}^{+\infty} \sum_{q=0}^{P-1} u_{\text{TX}}^{[q]}[\ell]  \, \psi_\text{DAC}(t- q \, T_\text{c}-\ell \, T)
\label{eq:xell}
\end{multline}
where $\pot_{\text{TX}}$ controls
the average transmission radio-frequency power
and we recall that $\psi_\text{DAC}(t)$
denotes the impulse response of the DAC.

The signal \eqref{eq:xell} is up-converted and transmitted
on the air.
Under A3, A4, and A5,
the baseband received signal
at the CU, after filtering, is given by
\begin{multline}
y(t) = \sum_{k=1}^{\Ka} A_{\text{U},k} \, e^{j \theta_{\text{U},k}} \, e^{j \, 2 \pi  f_{\text{U},k} t } \, \xa(t-\tau_{\text{U},k})
\\ +
\sum_{k=1}^{\Ks} A_{\text{J},k} \, e^{j \theta_{\text{J},k}} \, e^{j 2 \pi f_{\text{J},k} t } \, \xs(t-\tau_{\text{J},k}) + w(t)
\label{eq:sig}
\end{multline}
where
$K_\text{TX}$ is the number of paths between the generic transmitter
$\text{TX} \in \{\text{U},\text{J}\}$ and the CU,
$A_{\text{TX},k}>0$, $\theta_{\text{TX},k} \in[0, 2 \pi)$,
$f_{\text{TX},k}  \in [-D_\text{TX}/2, D_\text{TX}/2]$, and
$\tau_{\text{TX},k} \in [0, \Delta_\text{TX}]$
denote the amplitude, the phase shift, the Doppler
(frequency) shift, and delay of the $k$th path associated with
$x_\text{TX}(t)$, whereas
$w(t)$ is the complex envelope of (filtered) noise, which is statistically independent of $x_\text{TX}(t)$.
Without loss
of generality and according to A6, the LoS component
between the UAV and the CU
is assumed to be the first channel path, i.e.,
the one corresponding to $k=1$ in the leftmost sum in
\eqref{eq:sig}. Under A7,
the Doppler shifts
$f_{\text{U},1}, f_{\text{U},2}, \ldots, f_{\text{U},K_\text{U}}, f_{\text{J},1}, f_{\text{J},2}, \ldots, f_{\text{J},K_\text{J}}$ are pairwise-sum different.

The signal \eqref{eq:sig} is sampled with rate $1/T_\text{c}$ at instants
$t_{n,p} \eqdef n T + p \, T_\text{c}$,\footnote{Hereinafter, it is assumed that
$\tau_\text{min}=0$  without
loss of generality.} for $p \in \{0,1\ldots, P-1\}$.
Let  $y^{[p]}[n] \eqdef y(t_{n,p})$ be the discrete-time
counterpart of \eqref{eq:sig}, one gets
\begin{multline}
y^{[p]}[n] =
\sum_{k=1}^{\Ka} g_{\text{U},k} \, e^{j \, 2 \pi \nu_{\text{U},k} \left(n +\frac{p}{P} \right)  } \,
\xak^{[p]}[n]  \\ +
\sum_{k=1}^{\Ks} g_{\text{J},k} \, e^{j \, 2 \pi \nu_{\text{J},k} \left(n +\frac{p}{P}\right)} \,
\xsk^{[p]}[n] + w^{[p]}[n]
\label{eq:sig-ric}
\end{multline}
where
$g_{\text{TX},k} \eqdef \sqrt{2 \, \pot_{\text{TX}}} \, A_{\text{TX},k} \, e^{j \, \theta_{\text{TX},k}} \in \mathbb{C}$ is
referred to as the (complex) channel gain and
$\nu_{\text{TX},k} \eqdef f_{\text{TX},k} \, T
 \in [-1/2, 1/2)$ is the {\em normalized} Doppler shift,
whereas, for $\text{TX} \in \{\text{U},\text{J}\}$,
\begin{multline}
x_\text{TX,k}^{[p]}[n]=
\sum_{\ell=n-1}^{n} \sum_{q=0}^{P-1}
u_{\text{TX}}^{[q]}[\ell] \\
\alpha_{\text{TX},k}
\left[(n-\ell) P +(p-q) -d_{\text{TX},k}\right]
\label{eq:xpn}
\end{multline}
where, recalling A8,
\barr
\alpha_{\text{TX},k}[h] & \eqdef \psi(h \, T_\text{c}-\chi_{\text{TX},k})
\quad (h \in \mathbb{Z})
\\
\tau_{\text{TX},k} & = d_{\text{TX},k} \, T_\text{c} + \chi_{\text{TX},k}
\earr
with integer delay $d_{\text{TX},k} \in \left \{0,1, \ldots,
\lfloor \Delta_{\text{TX}}/T_\text{c}\rfloor \right\}$
and fractional delay $\chi_{\text{TX},k} \in [0, T_\text{c})$,
and, finally, $w^{[p]}[n] \eqdef w(t_{n,p})$.

By gathering the obtained samples of the received signal
into the vector
$\overline{\yb}[n] \eqdef (y^{[0]}[n], y^{[1]}[n],\ldots,y^{[P-1]}[n])^\trasp \in \Cset^P$ and accounting
for \eqref{eq:sig-ric}, we obtain the following vector model
\begin{multline}
\overline{\yb}[n] = \Hatildezero[n] \, \bab[n] + \Hatildeuno[n] \, \bab[n-1]
+
\Hstildezero[n] \, \bsb[n] \\ + \Hstildeuno[n] \, \bsb[n-1] + \overline{\wb}[n]
\label{eq:y-IBI}
\end{multline}
where, for $\text{TX} \in \{\text{U},\text{J}\}$ and $b \in \{0,1\}$,
\be
\overline{\mathbf{H}}_{\text{TX},b} [n] \eqdef \left[ \sum_{k=1}^{K_\text{TX}}
g_{\text{TX},k}  \, \mathbf{D}_{\text{TX},k} \,
\mathbf{T}_{\text{TX},k,b}  \, e^{j \, 2 \pi \nu_{\text{TX},k} n} \right]
\bm{\Omega}
\label{eq:mat}
\ee
with
\be
\mathbf{D}_{\text{TX},k} \eqdef \diag(1, e^{j \, \frac{2 \pi}{P} \nu_{\text{TX},k}},
\ldots, e^{j \, \frac{2 \pi}{P} \nu_{\text{TX},k} (P-1)})
\ee
whereas $\mathbf{T}_{\text{TX},k,b}\in \Rset^{P \times P}$ is a Toeplitz matrix
whose $(p+1,q+1)$th entry is
$\alpha_{\text{TX},k}\left[b\,  P +(p-q) -d_{\text{TX},k}\right]$,
for $p, q \in \{0,1\ldots, P-1\}$,
$\bm{\Omega} \eqdef \bm{I}_\text{cp} \, \bm{W}_M \in \Cset^{P \times M}$,
and, according to A9, the noise vector
$\overline{\wb}[n] \eqdef (w^{[0]}[n], w^{[1]}[n],\ldots,w^{[P-1]}[n])^\trasp \in \Cset^P$
is a zero-mean complex circular Gaussian random vector
with $\Es(\overline{\wb}[n_1] \, \overline{\wb}^\herm[n_2])=
\sigma_w^2 \, \delta_{n_1-n_2} \, \I_{P}$.

Under A10, 
after CP removal
the received vector \eqref{eq:y}
gets free from interblock interference (IBI) 
and assumes the form
\be
\yb[n] \eqdef \bm{R}_\text{cp} \, \overline{\yb}[n] = \Ha[n] \, \bab[n] +
\Hs[n] \, \bsb[n]  + \wb[n]
\label{eq:y}
\ee
where matrix $\bm{R}_\text{cp} \in \mathbb{R}^{M \times P}$ performs CP
removal,
$\bm{H}_{\text{TX}}[n] \eqdef \bm{R}_\text{cp} \,
\overline{\mathbf{H}}_{\text{TX},0} [n] \in \Cset^{M \times M}$,
and, finally,
the noise contribution is $\wb[n] \eqdef \bm{R}_\text{cp} \, \overline{\wb}[n] \in \Cset^M$.
It is noteworthy that both the UAV and the jammer transmissions are
adversely affected by
intercarrier interference (ICI) due to the presence of Doppler shifts.
Indeed, in the absence of a relative motion among the  two transmitters and
the CU, i.e., when $D_\text{U}=D_\text{J}=0$,
$\mathbf{H}_\text{TX} [n]$ boils down to a time-invariant circulant matrix
that can be diagonalized through DFT.

The main symbols introduced in Section~\ref{sec:system} are summarized in Table~\ref{tab:symbol}. 
In Section~\ref{sec:detection}, we assume that the CU has perfect knowledge of
both $\Ha[n]$ and $\Hs[n]$, by deferring
to Section~\ref{sec:estimation}
the challenging problem of how they can be
estimated from model \eqref{eq:y-IBI}.

\section{UAV symbol detection with anti-jamming capabilities}
\label{sec:detection}

The signal model \eqref{eq:y} describes a {\em multiple-access channel (MAC)}
\cite{Carl.1978, Wyner,Cover-Thomas-book,Tse}, where the CU receives signals from $2 M$
{\em virtual} transmitters, each with its own transmit symbol
$s_\text{TX}^{[p]}[n]$, for $\text{TX} \in \{\text{U}, \text{J}\}$
and $p \in \{0,1\ldots, M-1\}$.
The CU  processes the channel-output vector $\yb[n]$ on a symbol-by-symbol basis
({\em one-shot detection}) to detect the $M$ symbols transmitted by the UAV,
which are gathered in the vector $\bab[n]$.

Two types of strategies can be pursued to detect the UAV symbols.
In the former one, for each $p \in \{0,1\ldots, M-1\}$,
the CU detects the UAV symbol $s_\text{U}^{[p]}[n]$
by treating the remaining ICI symbols $s_\text{U}^{[0]}[n], \ldots, s_\text{U}^{[p-1]}[n],
s_\text{U}^{[p+1]}[n], \ldots, s_\text{U}^{[M-1]}[n]$
and jamming symbols $\bm{s}_\text{J}[n]$ as noise.
However, such a strategy leads to poor performance in practice,
since the power of ICI plus jamming contribution
in \eqref{eq:y} is comparable to (or, even, greater than) the
power of desired symbol.

In case of strong ICI and jamming,
significant performance gains
can be achieved by pursuing
a different strategy \cite{Carl.1978},
wherein the CU treats all the
inputs as a single vector
$\bm{s}[n] \eqdef (\bm{s}_\text{U}^\trasp[n],\bm{s}_\text{J}^\trasp[n])^\trasp \in \Cset^{2M}$
and all the symbols are detected jointly.
In this strategy, the estimate
of the symbol vector
$\bm{s}_\text{J}[n]$ -- whose information
content is useless to the CU --
allows one to eliminate the adverse effects
of the jamming symbols on UAV data detection.
Unfortunately, optimum (in the minimum-probability-error sense)
joint detection \cite{Proakis} entails
a computational complexity that grows
exponentially with the size of $\bm{s}[n]$, and, thus, becomes prohibitive
for a large number of subcarriers.
{\em Widely-linear (WL)} joint detection schemes \cite{Picinbono}, such as the zero-forcing
or minimum-mean-square-error (MMSE)
ones, allow one to
substantially reduce detection complexity
at the expense of (potentially
severe) ICI and jamming enhancement.

Herein, we propose the UAV symbol detection procedure with
anti-jamming capabilities depicted in Fig.~\ref{fig:fig_sic}.
We focus on the case where the transmitted symbols are real-valued, i.e.,
$\bm{s}[n]=\bm{s}^*[n]$ (the
developed detection approach can be straightforwardly extended to
complex noncircular constellations).
To exploit the noncircularity of
$\bm{s}[n]$, we build the following \textit{augmented} model for the received signal:
\barr
\z[n] \eqdef \begin{pmatrix}
\yb[n] \\
\yb^*[n]
\end{pmatrix}
& = \underbrace{\begin{pmatrix} \Ha[n] & \Hs[n] \\
\Ha^*[n] & \Hs^*[n] \end{pmatrix} }_{\Htilde[n] \in \Cset^{(2 M) \times (2 M)}}
 \bm{s}[n] + \underbrace{\begin{pmatrix}
\bm w[n] \\
\bm w^*[n]
\end{pmatrix}}_{\wtilde[n] \in \Cset^{2M}}
\nonumber \\ & = \Htilde[n] \, \bm{s}[n] + \wtilde[n] \: .
\label{eq:ytilde}
\earr
Dropping out the dependence on the block index $n$,
the detection procedure consists of $2M$ steps.
As reported in Fig.~\ref{fig:fig_sic},
at the $i$th step, with $i \in \{0,1,\ldots, 2M-1\}$, three distinct software-defined
functions
are executed:
\begin{enumerate}[(i)]

\item
WL-MMSE pre-detection;

\item
post-sorting algorithm (PSA);

\item
residual ICI-plus-jamming cancellation,
where the initial conditions
$\z_{0}=\z$, $\Htilde_{0}=\Htilde$, $\bm{s}_{0}=\bm{s}$ are set.

\end{enumerate}
The channel state information required by
such operations is provided by the
estimator discussed in Section~\ref{sec:estimation}.

\begin{figure}[t]
\centering
\includegraphics[width=\columnwidth]{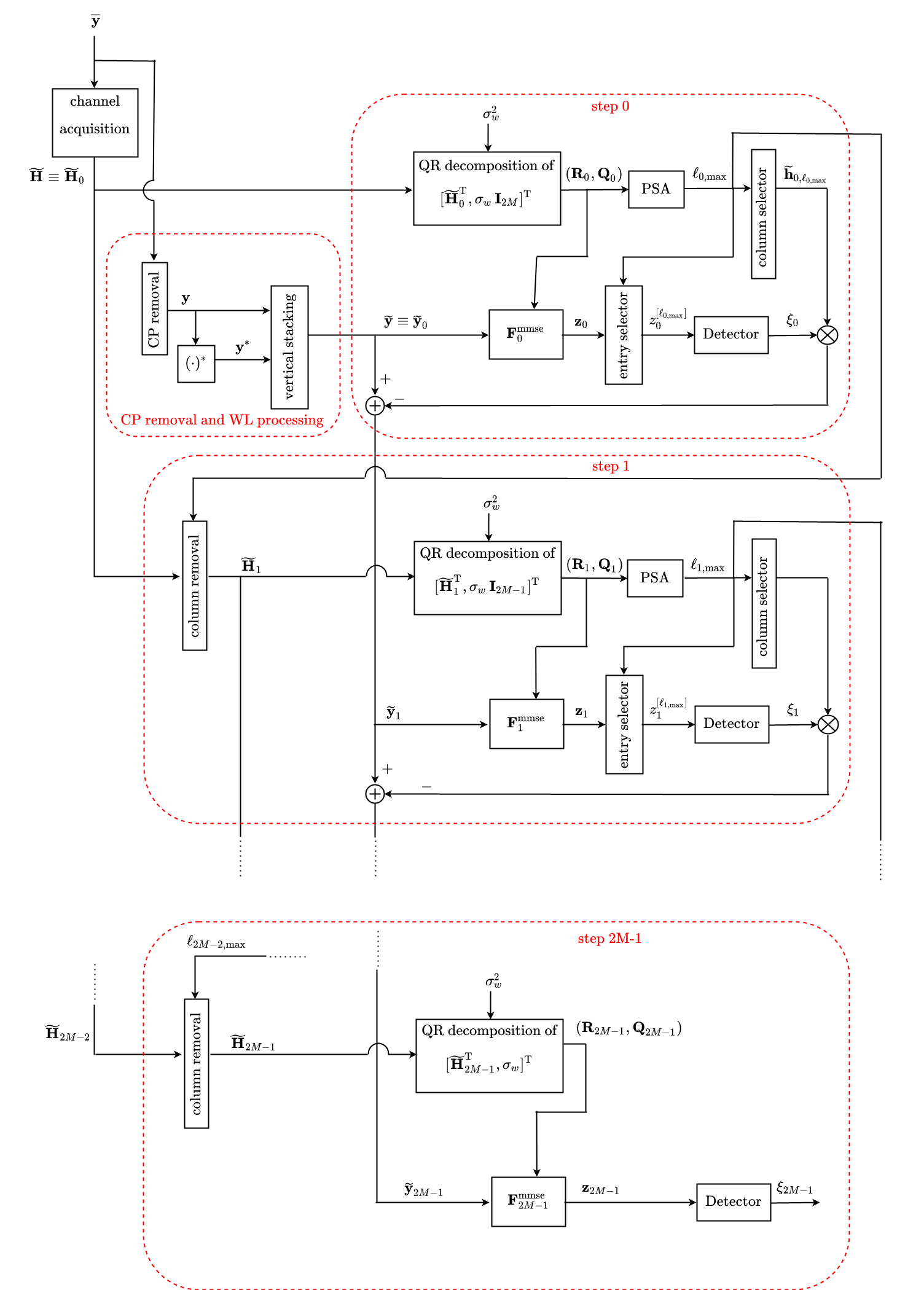}
\caption{The proposed anti-jamming reception structure.}
\label{fig:fig_sic}
\end{figure}

\subsection{WL-MMSE pre-detection}
\label{sec:WL-MMSE}

The WL-MMSE stage provides a preliminary ``soft'' estimate
$\bm{z}_i \eqdef (z_i^{[0]}, z_i^{[1]}, \ldots, z^{[2 M-i-1]}_k)^\trasp \in \Cset^{2 M-i}$
of the symbol vector at the $i$th step, for $i \in \{0,1,\ldots, 2M-1\}$.
In particular, the WL-MMSE  detector
\barr
\Fb^\text{mmse}_i & = \Htilde^\herm_i
\left(\Htilde_i \, \Htilde^\herm_i+ \sigma_w^2 \, \I_{2M}\right)^{-1}
\nonumber \\ & =
\left( \Htilde^\herm_i  \, \Htilde_i \,
+ \sigma_w^2 \, \I_{2M-i}\right)^{-1}
 \Htilde^\herm_i
\label{eq:WL-MMSE}
\earr
is obtained by minimizing
the mean squared error (MSE)
$\Es\left(\left\|\Fb_i \, \z_i - \bm{s}_i\right\|^2\right)$
between the output $\Fb_i \, \z_i$
of the WL detector $\Fb_i \in \Cset^{(2 M-i) \times (2 M)}$
and the symbol  vector $\bm{s}_i$, where,  for $i \in \{1,2,\ldots, 2M-1\}$,
$\z_i \in \Rset^{2M}$, $\bm{s}_i \in \Rset^{2M-i}$ and $\Htilde_i \in \Cset^{(2M) \times (2 M-i)}$
are obtained, respectively, from $\z_{i-1}$, $\bm{s}_{i-1}$, and
$\Htilde_{i-1}$ (see Subsection~\ref{sec:SICIJC}).
The output of the WL-MMSE
pre-detection stage is
\be
\bm{z}_i =   \Fb^\text{mmse}_i \, \z_i \: .
\label{eq:zk}
\ee
It is important to note that $\Fb^\text{mmse}_i \, \Htilde_i \neq \I_{2M-i}$,
that is, there are residual ICI and jamming contributions at the output of the WL-MMSE pre-detector.

An equivalent expression of
\eqref{eq:zk} can be obtained by observing that
\be
\Htilde^\herm_i  \, \Htilde_i + \sigma_w^2 \, \I_{2M-i}= \Hover_i^\herm \, \Hover_i =
\bm{R}_i^\herm \, \bm{R}_i
\label{eq:boh}
\ee
with $\Hover_i \eqdef (\Htilde_i^\trasp, \sigma_w \, \I_{2M-i})^\trasp \in \Cset^{(4M-i) \times (2M-i)}$ being full-column rank by construction, i.e.,
$\rank(\Hover_i )=2M-i$, whose economy-size QR decomposition \cite{Horn}
is $\Hover_i={\bm{Q}}_i \, {\bm{R}}_i$, where
${\bm{Q}}_i \in \Cset^{(4M-i) \times (2M-i)}$ is a semi-unitary matrix, i.e.,
${\bm{Q}}_i^\herm \, {\bm{Q}}_i=\I_{2M-i}$, and
${\bm{R}}_i \in \Cset^{(2M-i) \times (2M-i)}$ is an upper
triangular with nonzero diagonal elements. By substituting the QR decomposition
of $\Hover_i$  in \eqref{eq:WL-MMSE}, the pre-detector output \eqref{eq:zk}
can also be written as
\be
\bm{z}_i =
{\bm{R}}_i^{-1} \,
{\bm{Q}}_i^\herm
\begin{pmatrix}
\z_i \\
\bm{0}_{2M-i}
\end{pmatrix}
\label{eq:output-QR}
\ee
which is more suitable than \eqref{eq:zk} for
the subsequent PSA stage.

\vspace{2mm}
{\em Remark~1}:
When the WL-MMSE criterion is advocated for the soft estimation of
complex circular symbols, i.e., $\Es(\ba^2[n])=\Es(\bs^2[n])=0$, eq.~\eqref{eq:WL-MMSE} boils
down to a linear (L) MMSE filter, that is,
$\Fb^\text{mmse}_i=(\Fb^\text{l-mmse}_i, \Zero_{(2M-i) \times M})$,
where $\Fb^\text{l-mmse}_i \in \Cset^{(2M-i) \times M}$ is the standard
L-MMSE solution.

\subsection{Post-sorting algorithm (PSA)}
\label{sec:PSA}

The PSA is basically aimed at individuating
the soft symbol estimate
having the maximum
{\em signal-to-disturbance-plus-noise ratio (SDNR)}
(or, equivalently, that with the smallest estimation error)
at the output of the WL-MMSE filter, where
``disturbance'' stands for ICI-plus-jamming.
In order to introduce the PSA, we preliminarily  observe that
the {\em post-detection} SDNR $\gamma_{i,m}$ for
the $(m+1)$th entry of $\bm{s}_i$ is given \cite{Cioffi} by
\be
\gamma_{i,m} = \frac{1}{\sigma_w^2 \left\{\left( \Htilde^\herm_i  \, \Htilde_i \,
+ \sigma_w^2 \, \I_{2M-i}\right)^{-1}\right\}_{m+1,m+1}}-1
\label{eq:SDNR}
\ee
for $m \in \{0,1,\ldots, 2M-i-1\}$.
Expression \eqref{eq:SDNR} is based on the assumption of perfect knowledge of $\Htilde_i$.
On the other hand, when $\Htilde_i$ is replaced by its data-estimated version, eq.~\eqref{eq:SDNR} turns out to be particularly sensitive to finite-sample-size effects. 
A certain robustness against channel estimation errors can be
gained by resorting to the QR decomposition of $\Hover_i$ in eq.~\eqref{eq:boh}.
Since ${\bm{R}}_i$ is nonsingular and its inverse ${\bm{R}}_i^{-1}$
is another upper triangular matrix \cite{Horn}, one has $({\bm{R}}_i^\herm \,
{\bm{R}}_i)^{-1}={\bm{R}}_i^{-1} \, [{\bm{R}}_i^{-1}]^\herm $.
Therefore, the SDNR
\eqref{eq:SDNR} can be equivalently written as
\be
\gamma_{i,m} = \frac{1}{\displaystyle \sigma_w^2 \sum_{\ell=m}^{2M-i-1}
\left|\left\{{\bm{R}}_i^{-1}\right\}_{m+1,\ell+1}\right|^2}-1
\label{eq:SDNR-2}
\ee
for $m \in \{0,1,\ldots, 2M-i-1\}$.
We underline that \eqref{eq:SDNR} and \eqref{eq:SDNR-2}, as well as
\eqref{eq:zk} and \eqref{eq:output-QR},
are equivalent only when $\Htilde_i$
is perfectly known, whereas they give rise to
different values if $\Htilde_i$ is replaced by
its data-estimated counterpart.
We will resort to \eqref{eq:SDNR-2} for the PSA stage
since it is less sensitive to channel estimation errors.
The WL-MMSE soft estimate
with the largest SDNR corresponds to
the $(\ell_{i,\text{max}}+1)$th entry $z_{i}^{[\ell_{i,\text{max}}]}$ of $\bm{z}_i$,
with
\be
\ell_{i,\text{max}}  = \arg \max_{\ell \in \{0,1,\ldots, 2M-i-1\}} \gamma_{i, \ell} \:.
\label{eq:lmax}
\ee
The ``best'' soft estimate $z_i^{[\ell_{i,\text{max}}]}$
is then quantized to the nearest
(in Euclidean distance) information symbol to form the
``hard'' estimate $\xi_i=\mathcal{Q}(z_i^{[\ell_{i,\text{max}}]})$,
where $\mathcal{Q}(\cdot)$ is the minimum-distance
decision function. It is noteworthy that $\xi_i$ may be
either the estimate of a symbol transmitted by the UAV or
by the jammer.

\subsection{Residual ICI-plus-jamming cancellation}
\label{sec:SICIJC}

The cancellation
process  reduces the residual ICI and jamming
contributions so as to achieve better performance.
The basic idea is to cancel out the
contribution of $\xi_{i}$ to the remaining symbols,
for $i \in \{0,1,\ldots, 2M-1\}$. Starting from $\xi_{i}$, the
canceler executes the cancellation step
\be
\z_{i+1}=\z_{i}- \xi_{i} \, \htilde_{i,\ell_{i,\text{max}}}
\label{eq:canc}
\ee
where $\htilde_{i,\ell_{i,\text{max}}} \in \Cset^{2M}$
denotes the $(\ell_{i,\text{max}}+1)$th column of $\Htilde_{i}$.
In the case of correct detection, the updated
version of the received vector reads as
\be
\z_{i+1}=\Htilde_{i+1} \, \bm{s}_{i+1} + \wtilde
\ee
where $\Htilde_{i+1} \in \Cset^{(2M) \times (2M-i-1)}$
is obtained from $\Htilde_{i}$ by removing its
$(\ell_{i,\text{max}}+1)$th column $\htilde_{i,\ell_{i,\text{max}}}$ and, similarly,
$\bm{s}_{i+1} \in \Rset^{2M-i-1}$ is obtained from
$\bm{s}_{i}$ by removing its $(\ell_{i,\text{max}}+1)$th entry.

\begin{algorithm}[t]
\caption{The proposed UAV symbol detection procedure
with anti-jamming capabilities}
\label{table:tab_1}
\begin{algorithmic}
\STATE  {\bf Input quantities}: $\z_0=\z$, $\Htilde_0=\Htilde$,
and $\sigma_w^2$.

\STATE
{\bf Output quantities}: hard estimates $\xi_0, \xi_1, \ldots, \xi_{2M-1}$ of
the $2M$ transmitted symbols $\bm{s}_\text{U}$ and  $\bm{s}_\text{J}$.

\begin{enumerate}[1.]

\itemsep=1mm

\item
Set $i=0$.

\item
Evaluate the QR decomposition of $\Hover_i$.

\item
Perform the soft pre-detection \eqref{eq:output-QR}.

\item
Find the best soft estimate $z_i^{[\ell_{i,\text{max}}]}$ using \eqref{eq:lmax}.

\item
Compute the hard estimate $\xi_i=\mathcal{Q}(z_i^{[\ell_{k,\text{max}}]})$.

\item
Cancel out the symbol $\xi_i$ from $\z_i$ through \eqref{eq:canc}.

\item
Build $\Htilde_{i+1}$ from
$\Htilde_{i}$ by removing its
$(\ell_{i,\text{max}}+1)$th column.

\item
Set $i=i+1$: if $i < 2M$, then go back to Step 2, else end the procedure.

\end{enumerate}
\end{algorithmic}
\end{algorithm}

\subsection{Computational complexity}
\label{sec:complexity-detection}

The proposed UAV symbol detection procedure with
anti-jamming capabilities
is summarized as Algorithm~\ref{table:tab_1}
at the top of this page.
The computational complexity of Algorithm~\ref{table:tab_1} is mainly dominated by the QR decomposition of the
partitioned matrix  $\Hover_i = (\Htilde_i^\trasp, \sigma_w \, \I_{2M-i})^\trasp \in \Cset^{(4M-i) \times (2M-i)}$,
which involves
$\mathcal{O}[(4M-i) \, (2M-i)^2]$ floating point
operations \cite{Bjorck} if computed from scratch,
for $i \in \{0,1,\ldots, 2M-1\}$.
However, for
$i \in \{1,2,\ldots, 2M-1\}$, the matrix $\Hover_{i}$ is obtained
by deleting a row and a column from $\Hover_{i-1}$
(so-called {\em downdating} operations).
Algorithms for downdating the QR decomposition are
presented in \cite{Bjorck}: they can be modified in a straightforward
fashion to reduce the computational complexity of
Algorithm~\ref{table:tab_1},
by capitalizing on the fact
that the QR decomposition behaves nicely under partitioning.

To obtain an estimate of the blocks
of symbols transmitted  by the UAV during the
downloading phase,
one has to build the time-varying matrix
$\Htilde[n]$ in \eqref{eq:ytilde} and, then, implement
Algorithm~\ref{table:tab_1} for each symbol interval $n$
in the downloading time window.
Alternatively, one might resort to the {\em frequency-domain
representation} of such a linear time-variant filtering process \cite{Franks}, which
allows one to implement Algorithm~\ref{table:tab_1} by means of
linear time-invariant filtering
of frequency-shifted versions of the received vector
$\z[n]$ in \eqref{eq:ytilde}, thus reducing run-time complexity (in terms of floating point operations, subscripting, and memory traffic).

\section{Estimation of channel parameters}
\label{sec:estimation}

The basic assumption underlying the anti-jamming detection strategy developed in Section~\ref{sec:detection} is the knowledge of the channel matrices $\Ha[n]$ and $\Hs[n]$.
Acquisition of channel state information under a mobile jamming attack is a
challenging task, since channel estimation
in current wireless data communications standards
is mainly based
on training sequences sent
from the transmitter to the receiver.
Such training-based channel estimation algorithms
have been designed only to be resilient to
non-malicious interference and noise \cite{Fle.1999}.
A jammer can exploit this  weakness and efficiently
use its power budget in order to
significantly contaminate the
output of the channel estimator.

Henceforth, we propose a {\em blind} solution for estimating $\Ha[n]$ and $\Hs[n]$,
which does not rely on training sequences, but instead
exploits the unique
ACS properties
of both the UAV and jamming signals.
To this aim, for the sake of design,
we fix the number of paths
$K_\text{U}$ between the UAV and the CU
to a value large enough to
capture all the dominant impinging waves of the UAV transmission.\footnote{Alternatively, classical
information theoretic methods \cite{Wax.1985} for model selection
can be used to estimate the exact value of $K_\text{U}$ during a secure setup session.}
Instead of estimating the channels $\Ha[n]$ and $\Hs[n]$
({\em communication theory viewpoint}),
we choose herein to estimate
the parameters characterizing
such channels, i.e.,
$g_{\text{TX},k}$, $f_{\text{TX},k}$ or, equivalently,
$\nu_{\text{TX},k}$, and $\tau_{\text{TX},k}$, for
$k \in \{1,2,\ldots, K_\text{TX}\}$ and
$\text{TX} \in \{\text{U},\text{J}\}$,
by following an {\em array signal
processing approach}.

Specifically, we rely on the pioneering
works \cite{Chen.I, Chen.II},
wherein several reliable estimation methods
based on signal cyclostationarity have been developed
for time-difference-of-arrival of a single
signal of interest (SOI), such
as those based on cyclic cross-correlation,
spectral correlation ratio, and
spectral coherence alignment.
Estimation of  both
power and timing parameters of multiple
SOIs has been discussed in \cite{Gelli}.
However, the approaches in \cite{Chen.I, Chen.II,Gelli} are targeted
at flat-flat channels,
for which there exists a single path
only, i.e., $K_\text{U}=K_\text{J}=1$ in \eqref{eq:sig}, and the Doppler shift
can be assumed to be zero over the
observation time interval,
i.e., $f_{\text{U},1}=f_{\text{J},1}=0$ in \eqref{eq:sig}.
The basic methodologies developed in \cite{Chen.I, Chen.II,Gelli}
can be extended to time-varying channels,
by capitalizing on the fact
that Doppler shifts change
the spectral correlation property of the
transmitted signals \cite{Ciblat}, \cite[Ch.~8]{Napolitano.2019}.
However, it is assumed in such studies that the channel is
time selective {\em but} frequency nonselective (or frequency flat),
i.e., $K_\text{U}=K_\text{J}=1$ in \eqref{eq:sig},
which is a restrictive assumption for smart city applications,
where the channel might be non-flat in both time and frequency.

In the forthcoming subsections,
we exploit the ACS properties of the received signal
to blindly estimate channel gains, Doppler shifts, and time delays
of both the UAV and jamming transmissions, by generalizing
the methods developed in \cite{Chen.I, Chen.II,Gelli, Ciblat, Napolitano.2019}
to the case of doubly-selective channels.
The starting point of our estimation procedure is signal model
\eqref{eq:y-IBI}, i.e.,
the \textit{entire} OFDM block
(including the CP) is processed
for channel estimation purposes,
by also taking advantage of IBI.

\subsection{Second-order wide-sense statistical characterization of the received signal}

The ACS features of the received signal \eqref{eq:y-IBI}
stem from
the fact that $\overline{\mathbf{H}}_{\text{TX},b} [n]$
in \eqref{eq:mat} is a discrete-time
{\em almost-periodic} (AP) matrix \cite{Cord} with (possibly)
incommensurate frequencies belonging to the (order) set
\be
\mathcal{A}_\text{TX} \eqdef \{\nu_{\text{TX},1}, \nu_{\text{TX},2}, \ldots, \nu_{\text{TX},K_\text{TX}}\} \subseteq [-1/2, 1/2) \: .
\ee
As a consequence, for $r \in \{-1,0,1\}$,
the correlation matrix
$\Rb_{\overline{\yb}\overline{\yb}}[n, r] \eqdef \Es(\overline{\yb}[n] \,
\overline{\yb}^\herm[n-r]) \in \Cset^{P \times P}$
and the conjugate correlation matrix $\Rb_{\overline{\yb}\overline{\yb}^*}[n, r] \eqdef \Es(\overline{\yb}[n] \,
\overline{\yb}^\trasp[n-r]) \in \Cset^{P \times P}$
are both AP matrices.
Accordingly, the multivariate process \eqref{eq:y-IBI} is said to be
second-order wide-sense {\em almost-cyclostationary} \cite{Gardner-book}.
To save space, in the sequel, we use the notation
$\Rb_{\overline{\yb}\overline{\yb}^{(*)}}[n, r] \eqdef \Es(\overline{\yb}[n] \,
\{\overline{\yb}^\herm[n-r]\}^{(*)}) \in \Cset^{P \times P}$,
where the superscript $(*)$ denotes optional complex conjugation:
if conjugation is absent, then $\Rb_{\overline{\yb}\overline{\yb}^{(*)}}[n, r] \equiv \Rb_{\overline{\yb}\overline{\yb}}[n, r]$; if conjugation is present, then
$\Rb_{\overline{\yb}\overline{\yb}^{(*)}}[n, r] \equiv \Rb_{\overline{\yb}\overline{\yb}^{*}}[n, r]$.

It can be verified that
\begin{multline}
\Rb_{\overline{\yb}\overline{\yb}^{(*)}}[n, r] =
\sum_{k_1=1}^{\Ka} \sum_{h_1=1}^{\Ka}
\bm{\Xi}_{\text{U}^{(*)},k_1,h_1}[r] \, e^{j 2 \pi (\nu_{\text{U},k_1} -(-) \nu_{\text{U},h_1}) n}
\\ + \sum_{k_2=1}^{\Ks} \sum_{h_2=1}^{\Ks}
\bm{\Xi}_{\text{J}^{(*)},k_2,h_2}[r] \, e^{j 2 \pi (\nu_{\text{J},k_2} -(-) \nu_{\text{J},h_2}) n}
+ \Rb_{\overline{\wb}\overline{\wb}^{(*)}}[r]
\label{eq:Ryy}
\end{multline}
where $(-)$ is an optional minus sign linked to $(*)$: if conjugate is absent, then
$-(-)=-$; if conjugate is present, then
$-(-)=+$. Moreover,
the matrices in \eqref{eq:Ryy} are defined
in \eqref{eq:Psi-1}--\eqref{eq:Psi+1} at the top of the next page,
\begin{figure*}[!t]
\normalsize
\barr
\bm{\Xi}_{\text{TX}^{(*)},k,h}[-1]  & \eqdef  g_{\text{TX},k} \, (g_{\text{TX},h}^*)^{(*)}
\, (e^{- j \, 2 \pi \nu_{\text{TX},h}})^{(*)} \,
\mathbf{D}_{\text{TX},k} \, \mathbf{T}_{\text{TX},k,0} \, \bm{\Omega} \,
(\bm{\Omega}^\herm \, \mathbf{T}_{\text{TX},h,1}^\herm  \,
\mathbf{D}_{\text{TX},h}^*)^{(*)} \in \Cset^{P \times P}
\label{eq:Psi-1}
\\
\bm{\Xi}_{\text{TX}^{(*)},k,h}[0] & \eqdef
g_{\text{TX},k} \, (g_{\text{TX},h}^{*})^{(*)} \,
\mathbf{D}_{\text{TX},k} \, [\mathbf{T}_{\text{TX},k,0} \, \bm{\Omega} \,
(\bm{\Omega}^\herm \, \mathbf{T}_{\text{TX},h,0}^\herm)^{(*)}
+ \mathbf{T}_{\text{TX},k,1} \, \bm{\Omega} \,
(\bm{\Omega}^\herm \, \mathbf{T}_{\text{TX},h,1}^\herm)^{(*)}] \,
(\mathbf{D}_{\text{TX},h}^{*})^{(*)} \in \Cset^{P \times P}
\label{eq:Psi-0}
\\
\bm{\Xi}_{\text{TX}^{(*)},k,h}[1] & \eqdef
g_{\text{TX},k} \, (g_{\text{TX},h}^{*})^{(*)}
\, (e^{j \, 2 \pi \nu_{\text{TX},h}})^{(*)} \,
\mathbf{D}_{\text{TX},k} \, \mathbf{T}_{\text{TX},k,1} \, \bm{\Omega} \,
(\bm{\Omega}^\herm \, \mathbf{T}_{\text{TX},h,0}^\herm  \,
\mathbf{D}_{\text{TX},h}^{*})^{(*)} \in \Cset^{P \times P}
\label{eq:Psi+1}
\earr
\hrulefill
\end{figure*}
for $k,h \in \{1,2,\ldots, K_\text{TX}\}$ and
$\text{TX} \in \{\text{U},\text{J}\}$,
$\Rb_{\overline{\wb}\overline{\wb}^{(*)}}[r] \eqdef \Es(\overline{\wb}[n] \,
\{\overline{\wb}^\herm[n-r]\}^{(*)}) \in \Cset^{P \times P}$
is given by
$\Rb_{\overline{\wb}\overline{\wb}}[r] =\sigma_w^2 \, \I_{P} \, \delta_r$
when conjugation is absent,
whereas it boils down to
$ \Rb_{\overline{\wb}\overline{\wb}^{*}}[r] = \Zero_{P \times P}$, for
any $r \in \Zset$,  when conjugation is present.
The frequency set $\mathcal{A}_{\overline{\yb}\overline{\yb}^{(*)}}$
of the AP matrix \eqref{eq:Ryy}
can be written as reported in \eqref{eq:Ayy} at the top of the next page,
\begin{figure*}[!t]
\normalsize
\begin{multline}
\mathcal{A}_{\overline{\yb}\overline{\yb}^{(*)}} =
\Big \{\text{$\alpha_{\text{U}^{(*)},k_1,h_1} =\nu_{\text{U},k_1} -(-) \nu_{\text{U},h_1}$ (mod $1$),
for $k_1,h_1 \in \{1,2,\ldots, K_\text{U}\}$} \Big \}
\\ \cup
\Big \{\text{$\alpha_{\text{J}^{(*)},k_2,h_2} =\nu_{\text{J},k_2} -(-) \nu_{\text{J},h_2}$ (mod $1$),
for $k_2,h_2 \in \{1,2,\ldots, K_\text{J}\}$} \Big \}
\label{eq:Ayy}
\end{multline}
\hrulefill
\end{figure*}
which collects the pairwise {\em difference} of the Doppler shifts in $\mathcal{A}_\text{U}$ and
$\mathcal{A}_\text{J}$ in the case of the correlation matrix $\Rb_{\overline{\yb}\overline{\yb}}[n, r]$,
whereas it gathers the pairwise {\em sum} of the Doppler shifts in $\mathcal{A}_\text{U}$ and
$\mathcal{A}_\text{J}$ when the conjugate correlation matrix $\Rb_{\overline{\yb}\overline{\yb}^{*}}[n, r]$ is of concern.

Under mild regularity conditions, the AP matrix \eqref{eq:Ryy} can be equivalently written
in terms of the generalized Fourier series expansion \cite{Cord} as follows
\be
\Rb_{\overline{\yb}\overline{\yb}^{(*)}}[n,r] =
\sum_{\alpha \in \mathcal{A}_{\overline{\yb}\overline{\yb}^{(*)}}}
\Rb_{\overline{\yb}\overline{\yb}^{(*)}}^{\alpha}[r] \, e^{j 2 \pi \alpha n}
\ee
where, for $r \in \{-1,0,1\}$,
\barr
\Rb_{\overline{\yb}\overline{\yb}^{(*)}}^{\alpha}[r] & \eqdef
\left \langle \Rb_{\overline{\yb}\overline{\yb}^{(*)}}[n,r]
\, e^{-j 2 \pi \alpha n} \right \rangle
\nonumber \\ & =
\lim_{N \to +\infty}
\frac{1}{2 N+1} \sum_{n=-N}^{N} \Rb_{\overline{\yb}\overline{\yb}^{(*)}}[n,r]
\, e^{-j 2 \pi \alpha n}
\label{eq:Ralphagen}
\earr
is the $\alpha$th generalized Fourier series coefficient of the time-varying
correlation matrix $\Rb_{\overline{\yb}\overline{\yb}^{(*)}}[n,r]$, which is also
referred to as the {\em cyclic correlation matrix} (CCM) at cycle frequency $\alpha$
when conjugation is absent, whereas it is called the
{\em conjugate cyclic correlation matrix} (CCCM)
at cycle frequency $\alpha$
when conjugation is present.
It can be
verified that, for $r \in \{-1,0.1\}$,
\be
\Rb_{\overline{\yb}\overline{\yb}^{(*)}}^{\alpha}[r]  =
\begin{cases}
\bm{\Xi}_{\text{U}^{(*)},k_1,h_1}[r] \:, & \text{for $\alpha=\alpha_{\text{U}^{(*)},k_1,h_1}$}
\\ & \text{$k_1, h_1 \in \{1,2,\ldots, K_\text{U}\}$} \:;
\\
\bm{\Xi}_{\text{J}^{(*)},k_2,h_2}[r] \:, & \text{for $\alpha=\alpha_{\text{J}^{(*)},k_2,h_2}$}
\\ & \text{$k_2, h_2 \in \{1,2,\ldots, K_\text{J}\}$}\:;
\end{cases}
\label{eq:Ralpha}
\ee
whereas
$\Rb_{\overline{\yb}\overline{\yb}^{(*)}}^{\alpha}[r] \equiv \Zero_{P \times P}$
if $\alpha \not \in \mathcal{A}_{\overline{\yb}\overline{\yb}^{(*)}}$.\footnote{Eq.~\eqref{eq:Ralpha} holds in the more general case where the noise vector $\overline{\wb}[n]$ in \eqref{eq:y-IBI}
is a non-stationary colored random process that does not
exhibit almost-cyclostationarity at the cycle frequency
$\alpha \in \mathcal{A}_{\overline{\yb}\overline{\yb}^{(*)}}$.
}
The proposed estimation approaches rely only on the conjugate second-order statistics
of the almost-cyclostationary random process \eqref{eq:y-IBI}. In principle, the knowledge
of the
correlation matrices $\Rb_{\overline{\yb}\overline{\yb}}[n,-1]$,
$\Rb_{\overline{\yb}\overline{\yb}}[n,0]$, and $\Rb_{\overline{\yb}\overline{\yb}}[n,1]$
can also be exploited to improve estimation accuracy and/or relax some
assumptions.\footnote{\label{foot:relax}For instance, the assumption
that  the Doppler shifts are pairwise-sum different might be relaxed
by also using the set
$\mathcal{A}_{\overline{\yb}\overline{\yb}}$.
}
However, these enhancements are not pursued herein.

\vspace{2mm}
{\em Remark~2}: When the transmitted symbols are circular,
the conjugate correlation matrix
of the received signal \eqref{eq:y-IBI} is
identically zero, i.e.,  $\Rb_{\overline{\yb}\overline{\yb}^*}[n, r] \equiv \Zero_{P \times P}$.
In this case, as an alternative, the proposed
estimation algorithms can be suitably modified
to exploit \textit{higher-order cyclostationarity properties} of the
received signal \cite{Spooner}.

\subsection{Estimation of Doppler shifts}
\label{sec:Doppler}

Estimation
of Doppler shifts $f_{\text{TX},k}$
is equivalent to that of $\nu_{\text{TX},k}$,  for
$k \in \{1,2,\ldots, K_\text{TX}\}$ and
$\text{TX} \in \{\text{U},\text{J}\}$.
Such a knowledge can be acquired
by  relying on  the AP conjugate
correlation matrices
$\Rb_{\overline{\yb}\overline{\yb}^*}[n,-1]$,
$\Rb_{\overline{\yb}\overline{\yb}^*}[n,0]$, and $\Rb_{\overline{\yb}\overline{\yb}^*}[n,1]$
in \eqref{eq:Ryy},
where  the cardinality of
the frequency set $\mathcal{A}_{\overline{\yb}\overline{\yb}^*}$
defined in \eqref{eq:Ayy}
is equal to
\be
L_\text{a} \eqdef [\Ka \, (\Ka+1)+\Ks\, (\Ks+1)]/2
\label{eq:La}
\ee
since its elements
are different by assumption.
In the case of frequency-flat channels, i.e., $K_\text{U}=K_\text{J}=1$,
it follows that
$\mathcal{A}_{\overline{\yb}\overline{\yb}^*}=\left\{
2 \, \nu_{\text{U},1}, 2 \, \nu_{\text{J},1}\right\}$.
The complication due to the presence
of multipath stems from the fact that,
when $K_\text{U}> 1$ and $K_\text{J}>1$, the set
$\mathcal{A}_{\overline{\yb}\overline{\yb}^*}$ contains also 
the pairwise sums of
the elements of $\mathcal{A}_\text{U}$ and $\mathcal{A}_\text{J}$, which
significantly complicates
estimation of the individual
Doppler shift sets $\mathcal{A}_\text{U}$ and $\mathcal{A}_\text{J}$
from  the received data.

Herein, we first derive 
the proposed algorithm to estimate
the Doppler shifts under the assumption that $\Rb_{\overline{\yb}\overline{\yb}^*}^{\alpha}[-1]$, $\Rb_{\overline{\yb}\overline{\yb}^*}^{\alpha}[0]$, and
$\Rb_{\overline{\yb}\overline{\yb}^*}^{\alpha}[1]$ are perfectly known; the synthesis of its
data-estimated version
is briefly discussed at the end of this subsection.
Knowledge of the CCCMs $\Rb_{\overline{\yb}\overline{\yb}^*}^{\alpha}[-1]$,
$\Rb_{\overline{\yb}\overline{\yb}^*}^{\alpha}[0]$,
and $\Rb_{\overline{\yb}\overline{\yb}^*}^{\alpha}[1]$ enables one
to blindly retrieve
the unknown cycle frequencies
$\left\{\alpha_{\text{U}^{(*)},k_1,h_1}\right\}_{k_1,h_1=1}^{K_\text{U}}$
and
$\left\{\alpha_{\text{J}^{(*)},k_1,h_1}\right\}_{k_1,h_1=1}^{K_\text{J}}$
by defining the {\em one-dimensional} objective function:
\begin{multline}
J(\alpha) \eqdef \sum_{r=-1}^{1} \left \|\Rb_{\overline{\yb}\overline{\yb}^*}^{\alpha}[r] \right \|^2
\\ =
\sum_{r=-1}^{1} \trace\left(\Rb_{\overline{\yb}\overline{\yb}^*}^{\alpha}[r]
\left\{\Rb_{\overline{\yb}\overline{\yb}^*}^{\alpha}[r]\right\}^\herm \right)\: ,
\\ \quad \text{with $\alpha \in [-1/2, 1/2)$} \: .
\label{eq:J}
\end{multline}
By using differential calculus arguments,
it can be analytically shown
that $J(\alpha)$ has local maxima
in $[-1/2,1/2)$ at points
$\{\alpha_{\text{U}^{(*)},k_1,h_1}\}_{k_1,h_1=1}^{K_\text{U}}$
and $\{\alpha_{\text{J}^{(*)},k_2,h_2}\}_{k_2,h_2=1}^{K_\text{J}}$.
Hence, the cycle frequencies
of the second-order almost-cyclostationary
process \eqref{eq:y-IBI}  can be acquired by  searching
for the $L_\text{a}$ local maxima of \eqref{eq:J}
over the unit interval $[-1/2,1/2)$.\footnote{Acquisition
of $L_\text{a}$ can be obtained by counting the number of
local maxima of $J(\alpha)$ and/or by using
classical information
theoretic methods for model selection, such as Akaike's and
Rissanen's criteria  \cite{Wax.1985}.
}
Starting from the knowledge of $K_\text{U}$ and $L_\text{a}$,
and recalling \eqref{eq:La},
the number of paths $K_\text{J}$ between the jammer and the CU
can be obtained by taking the positive solution of the second-order equation
$K_\text{J}^2+K_\text{J}=2 \, L_\text{a}-\Ka \, (\Ka+1)$.

Let $\alphalabel \eqdef (\alpha_1, \alpha_2, \ldots, \alpha_{L_\text{a}})^\trasp \in [-1/2,1/2)^{L_\text{a}}$ collect the $L_\text{a}$ maximum points of  $J(\alpha)$ in the interval
$[-1/2, 1/2)$, the next step
is to infer the Doppler shifts gathered into the
two vectors
$\boldsymbol{\nu}_\text{U} \eqdef (\nu_{\text{U},1}, \nu_{\text{U},2}, \ldots,
\nu_{\text{U},K_\text{U}})^\trasp \in [-1/2,1/2)^{\Ka}$
and
$\boldsymbol{\nu}_\text{J} \eqdef (\nu_{\text{J},1}, \nu_{\text{J},2}, \ldots,
\nu_{\text{J},K_\text{J}})^\trasp \in [-1/2,1/2)^{\Ks}$
from the knowledge of $\alphalabel$. Let $\boldsymbol{\nu}_\text{ord} \eqdef (\boldsymbol{\nu}_\text{U}^\trasp,
\boldsymbol{\nu}_\text{J}^\trasp)^\trasp \in [-1/2,1/2)^{K}$.

\vspace{2mm}
\begin{proposition}
\label{prop:1}
Estimation of $\boldsymbol{\nu}_\text{ord}$
from the observed vector
$\alphalabel$ is tantamount to
solving the matrix equation
\be
\Pblabel \, \bm B \, \boldsymbol{\nu}_\text{ord} = \alphalabel
\label{eq:syst}
\ee
with respect to the unknowns $\Pblabel$ and $\boldsymbol{\nu}_\text{ord}$,
where we have defined the matrix
$\bm B  \eqdef (2 \, \bm I_{K}, \bm \Gamma^\trasp)^\trasp \in \Rset^{L_\text{a} \times K}$,
with
\be
\bm \Gamma \eqdef \left( \begin{array}{ccccc} 1& 1& 0& \ldots & 0 \\
1& 0& 1& \ldots & 0 \\
\vdots & \ldots & \ddots & \ddots & \vdots \\
1 & 0 & 0 & \ldots & 1 \\
0 & 1 & 1 & \ldots & 0 \\
\vdots & \ldots & \ddots & \ddots & \vdots \\
\end{array} \right)  \in \mathbb{R}^{(L_\text{a}-K) \times K}
\label{eq:Gamma}
\ee
whereas $\Pblabel \in \Rset^{L_\text{a} \times L_\text{a}}$ is a permutation matrix \cite{Horn}
having exactly one entry in each row and column equal to $1$, and
all the other entries equal to $0$.

\end{proposition}

\begin{proof}
See Appendix~\ref{app:1}.
\end{proof}

\vspace{2mm}
The permutation matrix $\Pblabel$ is completely characterized by
a vector of integers, where the integer at position $i$ is
the column index of the unit element of row $i$ of $\Pblabel$.
Hence, eq.~\eqref{eq:syst} represents a {\em mixed-integer}
system of $L_\text{a}$ linear equations.

System~\eqref{eq:syst} is underdetermined in the unknowns $\Pblabel$  and
$\boldsymbol{\nu}_\text{ord}$,
and can be solved in the
least-squares (LS) sense (see, e.g., \cite{Kay})
by formulating the
following constrained problem:
\be
\arg \min_{\shortstack{\footnotesize
$\Pb \in \mathcal{P}$ \\ \footnotesize $\boldsymbol{\nu} \in [-1/2,1/2)^{K}$}}
\left \|\Pb \, \bm B \, \boldsymbol{\nu}- \alphalabel \right \|^2 \:
\label{eq:LS}
\ee
where $\mathcal{P}$ denote the subset of all $L_\text{a} \times L_\text{a}$
permutation matrices.

For a given permutation matrix $\Pb$, minimization in \eqref{eq:LS} 
with respect to $\boldsymbol{\nu}$ is an overdetermined
LS problem, whose solution is given by
$\boldsymbol{\nu}_\text{LS} = \left( {\bm B}^\trasp\, \bm B \right)^{-1} {\bm B}^\trasp \,
\Pb^\trasp \, \alphalabel$ \cite{Kay}.
Substituting such a LS solution into \eqref{eq:LS}, we obtain
the {\em integer} linear programming problem:
\be
\arg \max_{\Pb \in \mathcal{P}}
\left( \alphalabel^\trasp \, \Pb \, \boldsymbol{\Pi}
\, \Pb^\trasp \, \alphalabel \right)
\label{eq:LSmax}
\ee
where $\boldsymbol{\Pi} \eqdef {\bm B} \left( {\bm B}^\trasp\, \bm B \right)^{-1} {\bm B}^\trasp
\in \Rset^{L_\text{a} \times L_\text{a}}$
is the orthogonal projector onto the range (or column space)  of $\bm{B}$.
Since $\Pb^\trasp=\Pb^{-1}$ permutes columns in the same way as the
permutation matrix $\Pb$ permutes rows, the transformation
$\boldsymbol{\Pi} \rightarrow \Pb \, \boldsymbol{\Pi}
\, \Pb^\trasp$ permutes the rows and the columns of $\boldsymbol{\Pi}$
in the same way.
By virtue of the Rayleigh-Ritz theorem \cite{Horn}, one has
\begin{multline}
\alphalabel^\trasp \, \Pb \, \boldsymbol{\Pi}
\, \Pb^\trasp \, \alphalabel \le \lambda_\text{max}(\Pb \, \boldsymbol{\Pi}
\, \Pb^\trasp) \, \|\alphalabel\|^2 \\ =\lambda_\text{max}(\boldsymbol{\Pi}) \, \|\alphalabel\|^2
= \|\alphalabel\|^2=\|\boldsymbol{\alpha}_\text{ord}\|^2
\end{multline}
where we have used the facts that the transformation $\boldsymbol{\Pi} \rightarrow \Pb \, \boldsymbol{\Pi}
\, \Pb^\trasp$ does not alter the eigenvalues of $\boldsymbol{\Pi}$ and, additionally, the eigenvalues
of an orthogonal projector matrix are $0$ or $1$.
Taking into account \eqref{eq:L1}--\eqref{eq:L2} and remembering
that $\Pblabel^\trasp \, \Pblabel=\I_{L_\text{a}}$,
it is readily verified that the cost function in \eqref{eq:LSmax} is maximized when $\Pb=\Pblabel$.
However, it can be numerically verified that, owing to the
particular structure of matrix $\bm B$ in \eqref{eq:sys-1},
there exist permutation matrices $\Pb \neq \Pblabel$
such that $\Pb \, \boldsymbol{\Pi}
\, \Pb^\trasp =\Pblabel \, \boldsymbol{\Pi}
\, \Pblabel^\trasp$.
In a nutshell, the optimization problem \eqref{eq:LSmax} admits
more than one solution:
this result mathematically certifies the fact that the composite mapping
$\mathcal{L}_1 \circ \mathcal{L}_2$ is not invertible, as previously announced.

Let  $\Pb_{\text{LS},0}, \Pb_{\text{LS},1}, \ldots, \Pb_{\text{LS},Q-1}$
denote the $Q$ permutation matrices achieving the same optimal value $\|\alphalabel\|^2$
in \eqref{eq:LSmax}, where $\Pb_{\text{LS},0} \eqdef \Pblabel$ and
$Q$ can be upper bounded as $Q < K!$, and
consider the corresponding estimates of the Doppler shifts
\begin{multline}
\boldsymbol{\nu}_{\text{LS},q} = \left( {\bm B}^\trasp\, \bm B \right)^{-1} {\bm B}^\trasp \,
\Pb_{\text{LS},q}^\trasp \, \alphalabel \: ,
\\ \quad \text{for $q \in \{0,1,\ldots, Q-1\}$} \: .
\label{eq:solfin}
\end{multline}
It can be observed by direct inspection that
$\boldsymbol{\nu}_{\text{LS},0}=\boldsymbol{\nu}_\text{ord}$,
whereas
$\boldsymbol{\nu}_{\text{LS},1}, \boldsymbol{\nu}_{\text{LS},2}, \ldots,
\boldsymbol{\nu}_{\text{LS},Q-1}$ are different permutations of the desired
vector $\boldsymbol{\nu}_\text{ord}$.
This inherent {\em permutation ambiguity} has to be solved in order
to correctly distinguish the Doppler shifts of the UAV-to-CU
transmission from those of the jammer-to-CU one.
Remarkably, it is not necessary to completely remove such an ambiguity for
detection purposes by singling $\boldsymbol{\nu}_\text{ord}$ out
in $\{\boldsymbol{\nu}_{\text{LS},q}\}_{q=0}^{Q-1}$, but it is instead sufficient to
separate the common entries of the
vectors $\{\boldsymbol{\nu}_{\text{LS},q}\}_{q=0}^{Q-1}$ in two sets
$\widetilde{\mathcal{A}}_\text{U}$ and $\widetilde{\mathcal{A}}_\text{J}$:
the former one gathers all the Doppler shifts of the UAV-to-CU transmission
in any order ($\widetilde{\mathcal{A}}_\text{U}$ is
the unordered version of $\mathcal{A}_\text{U}$);
the latter one collects all the Doppler shifts of the jammer-to-CU transmission
in any order ($\widetilde{\mathcal{A}}_\text{J}$ is
the unordered version of $\mathcal{A}_\text{J}$).
\begin{algorithm}[t]
\caption{The proposed algorithm for retrieval of Doppler shifts}
\label{table:tab_2}
\begin{algorithmic}
\STATE  {\bf Input quantities}: $K_\text{U}$, $\nu_{\text{U},1}$,
$\Rb_{\overline{\yb}\overline{\yb}^*}^{\alpha}[-1]$,
$\Rb_{\overline{\yb}\overline{\yb}^*}^{\alpha}[0]$, and $\Rb_{\overline{\yb}\overline{\yb}^*}^{\alpha}[1]$.

\STATE
{\bf Output quantities}: $\widetilde{\mathcal{A}}_\text{U}$ and
$\widetilde{\mathcal{A}}_\text{J}$.

\STATE
{\bf Initialization}: $i=1$ and $\widetilde{\mathcal{A}}_\text{U}=\emptyset$.

\begin{enumerate}[1.]

\itemsep=1mm

\item
Search for the $L_\text{a}$ local maxima of \eqref{eq:J}
over $[-1/2,1/2)$
and collect them in
$\alphalabel$.

\item
Evaluate the number of paths $K_\text{J}$ between the jammer and the CU
by taking the positive solution of
$K_\text{J}^2+K_\text{J}=2 \, L_\text{a}-\Ka \, (\Ka+1)$.

\item
Generate a tentative permutation matrix $\Pb \in \Rset^{L_\text{a} \times L_\text{a}}$.

\item
{\bf If} $\alphalabel^\trasp \, \Pb \, \boldsymbol{\Pi}
\, \Pb^\trasp \, \alphalabel = \|\alphalabel\|^2$, {\bf then} $\Pb_\text{LS}=\Pb$
is a solution of \eqref{eq:LSmax} and
go on Step 5, {\bf else} go back to Step 3 by generating a different tentative permutation matrix.

\item
Compute the vector of Doppler shifts $\boldsymbol{\nu}_{\text{LS}} = \left( {\bm B}^\trasp\, \bm B \right)^{-1} {\bm B}^\trasp \,  \Pb_\text{LS}^\trasp \, \alphalabel$.

\WHILE{cardinality of $\widetilde{\mathcal{A}}_\text{U}$ is less than $K_\text{U}$}
        \STATE {{\bf if} $\{\boldsymbol{\nu}_{\text{LS}}\}_i
        +\nu_{\text{U},1}$ is an entry of $\alphalabel$

        {\bf then}

        add $\{\boldsymbol{\nu}_{\text{LS}}\}_i$ to $\widetilde{\mathcal{A}}_\text{U}$ and set $i=i+1$.
}

 \ENDWHILE

\item
$\widetilde{\mathcal{A}}_\text{J}$ is the set of elements of $\boldsymbol{\nu}_{\text{LS}}$ not in
$\widetilde{\mathcal{A}}_\text{U}$.

\end{enumerate}
\end{algorithmic}
\end{algorithm}
This classification problem can be solved by resorting to Algorithm~\ref{table:tab_2}
reported at the top of this page, which is based on the additional knowledge of the Doppler shift $f_{\text{U},1}$ (or, equivalently, $\nu_{\text{U},1}$) of the LoS path between the UAV
and the CU.\footnote{\label{foot:1}Since floating point calculations are carried out
with a finite precision depending on the used compilers and CPU architectures, the comparison between
$\{\boldsymbol{\nu}_{\text{LS}}\}_i +\nu_{\text{U},1}$ and the $\ell$th entry $\{\alphalabel\}_\ell$ of $\alphalabel$ in
Algorithm~\ref{table:tab_2} has to be numerically done by checking if
the relative error  $\left|\{\boldsymbol{\nu}_{\text{LS}}\}_i
+\nu_{\text{U},1}-\{\alphalabel\}_\ell\right|/\left|\{\alphalabel\}_\ell\right|$
is smaller than a sufficiently small threshold $\epsilon>0$.}
The Doppler shift of such a LoS component depends on the relative speed, the transmitting
wavelength, and the moving direction of the UAV with respect to the CU.
Such parameters are commonly measured and/or estimated by the UAV relying on
a set of sensors in conjunction with appropriate on-board algorithms  and they are typically communicated
to the CU through a secure control  and non-payload communication link \cite{Zeng}.

A key observation underlying our blind estimation approach
is that, in practice, $\Rb_{\overline{\yb}\overline{\yb}^*}^{\alpha}[r]$
can be directly estimated from the received data
using the consistent estimate (see, e.g., \cite{Dand})
\be
\widehat{\Rb}_{\overline{\yb}\overline{\yb}^*}^{\alpha}[r] =
\frac{1}{2 N+1} \sum_{n=-N}^{N}
\overline{\yb}[n] \,  \overline{\yb}^\trasp[n-r] \, e^{-j \, 2 \pi \alpha n}
\label{eq:Rciclest}
\ee
for $r \in \{-1,0,1\}$. Hence, in practice, an estimate $\alphalabelhat$ of the vector
$\alphalabel$ can be directly obtained from data by finding the $\widehat{L}_\text{a}$
distinct peaks of the function
\begin{multline}
\widehat{J}(\alpha) \eqdef \sum_{r=-1}^{1} \left \|\widehat{\Rb}_{\overline{\yb}\overline{\yb}^*}^{\alpha}[r] \right \|^2 \\ =
\sum_{r=-1}^{1} \trace\left(\widehat{\Rb}_{\overline{\yb}\overline{\yb}^*}^{\alpha}[r]
\left\{\widehat{\Rb}_{\overline{\yb}\overline{\yb}^*}^{\alpha}[r]\right\}^\herm \right)
\\ \quad \text{with $\alpha \in [-1/2, 1/2)$}
\label{eq:JN}
\end{multline}
which is the starting point to build the estimates
$\widehat{\mathcal{A}}_\text{U}$ and $\widehat{\mathcal{A}}_\text{J}$
of the  Doppler shift sets
$\widetilde{\mathcal{A}}_\text{J}$ and $\widetilde{\mathcal{A}}_\text{J}$, respectively.
The estimation procedure is a straightforward modification of
Algorithm~\ref{table:tab_2}, with $J(\alpha)$, $L_\text{a}$, $\alphalabel$, and
$\boldsymbol{\nu}_{\text{LS}}$ 
replaced with $\widehat{J}(\alpha)$, $\widehat{L}_\text{a}$,
$\alphalabelhat$, and
$\widehat{\boldsymbol{\nu}}_{\text{LS}}=\left( {\bm B}^\trasp\, \bm B \right)^{-1} {\bm B}^\trasp \,  \Pb_\text{LS}^\trasp \, \alphalabelhat$, respectively, and the comparison
between $\{\widehat{\boldsymbol{\nu}}_{\text{LS}}\}_i +\nu_{\text{U},1}$ and the $\ell$th entry
$\{\alphalabelhat\}_\ell$ of $\alphalabelhat$ is performed 
by checking the corresponding relative error
(see footnote~\ref{foot:1}).
The consistency and asymptotic normality
of estimate for $\widehat{J}(\alpha)$ can be proven by paralleling
the derivations detailed in \cite{Ciblat}.

\begin{algorithm}[t]
\caption{The proposed algorithm for acquisition of time delays and channel gains}
\label{table:tab_3}
\begin{algorithmic}
\STATE  {\bf Input quantities}: $\widetilde{\mathcal{A}}_\text{U}$,
$\widetilde{\mathcal{A}}_\text{J}$, $\Rb_{\overline{\yb}\overline{\yb}^*}^{\alpha}[-1]$,
$\Rb_{\overline{\yb}\overline{\yb}^*}^{\alpha}[0]$, and $\Rb_{\overline{\yb}\overline{\yb}^*}^{\alpha}[1]$.

\STATE
{\bf Output quantities}: $\tau_{\text{TX},k}$ and $g_{\text{TX},k}$,
for each $k \in \{1,2,\ldots, K_\text{TX}\}$ and
$\text{TX} \in \{\text{U},\text{J}\}$.

\begin{enumerate}[1.]

\itemsep=1mm

\item
Build the matrices $\bm{\Phi}_{\text{TX},k}[-1] $, $\bm{\Phi}_{\text{TX},k}[0]$,
and $\bm{\Phi}_{\text{TX},k}[1] $ as defined in
\eqref{eq:Phi-1}-\eqref{eq:Phi+1}.

\item
Construct the matrix $\bm{\Phi}_{\text{TX},k}=\bm{\Phi}_{\text{TX},k}[-1] +
\bm{\Phi}_{\text{TX},k}[0]+ \bm{\Phi}_{\text{TX},k}[1] $ as reported in
\eqref{eq:Phi}.

\item
The time delay $\tau_{\text{TX},k}$ is acquired by
search for the maximum of \eqref{eq:I}
over $\beta \in [0, \Delta_\text{max}]$.

\item
Use $\tau_{\text{TX},k}$ to
form the matrix
$\mathbf{C}_{\text{TX},k}=\bm{W}_P \,  \diag(\vb_{\text{TX},k}) \,
\bm{W}_P^\herm$.

\item
Channel gains are obtained from \eqref{eq:g}
up to a sign.

\end{enumerate}
\end{algorithmic}
\end{algorithm}

\subsection{Joint estimation of channel gains and time delays}
\label{sec:Doppler}
Hereinafter, we assume that the sets $\widetilde{\mathcal{A}}_\text{U}$ and
$\widetilde{\mathcal{A}}_\text{J}$ are perfectly known.
The proposed technique to jointly estimate the channel gains
$g_{\text{TX},k}$ and delays
$\tau_{\text{TX},k}= d_{\text{TX},k} \, T_\text{c} + \chi_{\text{TX},k}$,
for any $k \in \{1,2,\ldots, K_\text{TX}\}$ and
$\text{TX} \in \{\text{U},\text{J}\}$,
also relies on  the AP
conjugate correlation matrices
$\Rb_{\overline{\yb}\overline{\yb}^*}[n,-1]$,
$\Rb_{\overline{\yb}\overline{\yb}^*}[n,0]$, and $\Rb_{\overline{\yb}\overline{\yb}^*}[n,1]$
in \eqref{eq:Ryy}.
Specifically,
let us consider the CCCMs \eqref{eq:Ralpha}, for $r \in\{-1,0,1\}$,
at cycle frequency $\alpha=\alpha_{\text{TX}^*,k,k}=2 \, \nu_{\text{TX},k}$, which,
according to \eqref{eq:Psi-1}--\eqref{eq:Psi+1}, assume the forms
reported in \eqref{eq:Ryyalpha-1}--\eqref{eq:Ryyalpha+1} at the top of the next page,
\begin{figure*}[!t]
\normalsize
\barr
\Rb_{\overline{\yb}\overline{\yb}^*}^{2 \, \nu_{\text{TX},k}}[-1]
& = g_{\text{TX},k}^2
\, e^{j \, 2 \pi \nu_{\text{TX},k}} \,
\mathbf{D}_{\text{TX},k} \, \mathbf{T}_{\text{TX},k,0} \,
\bm{\Omega} \, \bm{\Omega}^\trasp \,
\mathbf{T}_{\text{TX},k,1}^\trasp  \,
\mathbf{D}_{\text{TX},k}
\label{eq:Ryyalpha-1}
\\
\Rb_{\overline{\yb}\overline{\yb}^*}^{2 \, \nu_{\text{TX},k}}[0] & =
g_{\text{TX},k}^2 \,
\mathbf{D}_{\text{TX},k} \, \left(\mathbf{T}_{\text{TX},k,0} \,
\bm{\Omega} \, \bm{\Omega}^\trasp
\, \mathbf{T}_{\text{TX},k,0}^\trasp+
\mathbf{T}_{\text{TX},k,1} \,
\bm{\Omega} \, \bm{\Omega}^\trasp \,
\mathbf{T}_{\text{TX},k,1}^\trasp \right)
\mathbf{D}_{\text{TX},k}
\label{eq:Ryyalpha-0}
\\
\Rb_{\overline{\yb}\overline{\yb}^*}^{2 \, \nu_{\text{TX},k}}[1] & =
g_{\text{TX},k}^2
\, e^{-j \, 2 \pi \nu_{\text{TX},h}} \,
\mathbf{D}_{\text{TX},k} \, \mathbf{T}_{\text{TX},k,1} \,
\bm{\Omega} \, \bm{\Omega}^\trasp \,
\mathbf{T}_{\text{TX},k,0}^\trasp  \,
\mathbf{D}_{\text{TX},k}
\label{eq:Ryyalpha+1}
\earr
\hrulefill
\end{figure*}
for $k \in \{1,2,\ldots, K_\text{TX}\}$ and $\text{TX} \in \{\text{U},\text{J}\}$.
Relying on \eqref{eq:Ryyalpha-1}-\eqref{eq:Ryyalpha+1} and the previously
acquired knowledge of the Doppler shifts,
let us build the matrices
\barr
\bm{\Phi}_{\text{TX},k}[-1] & \eqdef  e^{-j 2 \pi \nu_{\text{TX},k}}  \,
\mathbf{D}_{\text{TX},k}^* \, \Rb_{\overline{\yb}\overline{\yb}^*}^{2 \,
\nu_{\text{TX},k}}[-1] \, \mathbf{D}_{\text{TX},k}^*
\nonumber \\ & =
g_{\text{TX},k}^2 \, \mathbf{T}_{\text{TX},k,0} \,
\bm{\Omega} \, \bm{\Omega}^\trasp \,
\mathbf{T}_{\text{TX},k,1}^\trasp
\label{eq:Phi-1}
\\
\bm{\Phi}_{\text{TX},k}[0] & \eqdef
\mathbf{D}_{\text{TX},k}^* \, \Rb_{\overline{\yb}\overline{\yb}^*}^{2 \, \nu_{\text{TX},k}}[0]
\, \mathbf{D}_{\text{TX},k}^*
\nonumber \\ & =
g_{\text{TX},k}^2  \left(\mathbf{T}_{\text{TX},k,0} \,
\bm{\Omega} \, \bm{\Omega}^\trasp
\, \mathbf{T}_{\text{TX},k,0}^\trasp
\right. \nonumber \\  & + \left.
\mathbf{T}_{\text{TX},k,1} \,
\bm{\Omega} \, \bm{\Omega}^\trasp \,
\mathbf{T}_{\text{TX},k,1}^\trasp \right)
\label{eq:Phi0}
\\
\bm{\Phi}_{\text{TX},k}[1] & \eqdef
e^{j  2 \pi \nu_{\text{TX},k}} \, \mathbf{D}_{\text{TX},k}^* \, \Rb_{\overline{\yb}\overline{\yb}^*}^{2 \, \nu_{\text{TX},k}}[1] \, \mathbf{D}_{\text{TX},k}^*
\nonumber \\ & =
g_{\text{TX},k}^2 \,  \mathbf{T}_{\text{TX},k,1} \,
\bm{\Omega} \, \bm{\Omega}^\trasp \,
\mathbf{T}_{\text{TX},k,0}^\trasp  \: .
\label{eq:Phi+1}
\earr

\vspace{2mm}
\begin{proposition}
\label{prop:2}
Acquisition of the delay
$\tau_{\text{TX},k}$ can be pursued by searching over the
interval $[0, \Delta_\text{max}] \subset [0, T/2)$
for the global maximum of the {\em one-dimensional} cost function defined in
\eqref{eq:I} at the top of the next page,
\begin{figure*}[!t]
\normalsize
\barr
\mathcal{I}_{\text{TX},k}(\beta) & \eqdef
\left| \sum_{p=0}^{P/2-1} \{\bm{W}_P^\herm \, \bm{\Phi}_{\text{TX},k} \, \bm{W}_P^*\}_{p,p}
\, (\Psi_{\text{TX},k}^2[p])^*  \, \{\bm{\Upsilon}\}_{p,p}^* \, e^{j \frac{4 \pi}{T} \beta \, p}
\right|
\nonumber \\ & = |g_{\text{TX},k}|^2
\left| \sum_{p=0}^{P/2-1} \left|\Psi_{\text{TX},k}[p]\right|^4
\, \left|\{\bm{\Upsilon}\}_{p,p}\right|^2 e^{-j \frac{4 \pi}{T} (\tau_{\text{TX},k}-\beta) \, p}
\right|
\:, \quad \text{with $\beta \in [0, \Delta_\text{max}]$}
\label{eq:I}
\earr
\hrulefill
\end{figure*}
where
\barr
\bm{\Phi}_{\text{TX},k} & \eqdef \sum_{r=-1}^{1} \bm{\Phi}_{\text{TX},k}[r]
\\
\bm{\Upsilon} & \eqdef \bm{W}_P^\herm \, \bm{\Omega} \, \bm{\Omega}^\trasp \,
\bm{W}_P^*  \in \Cset^{P \times P}
\label{eq:ups}
\earr
and, for $p \in \{0,1,\ldots, P-1\}$,
the sample $\Psi_{\text{TX},k}[p]$ is given by \eqref{eq:58}--\eqref{eq:59}
in Appendix~\ref{app:3}.
\end{proposition}

\begin{proof}
See Appendix~\ref{app:2}.
\end{proof}

\vspace{2mm}
Once the Doppler shifts and delays are acquired, the fading coefficients
can be estimated  from \eqref{eq:Phi} by solving the LS optimization problem:
\be
\arg \min_{\rho \in \Cset}
\left\|\bm{\Phi}_{\text{TX},k} -
\rho^2 \, \mathbf{C}_{\text{TX},k} \, \bm{\Omega} \, \bm{\Omega}^\trasp \,
\mathbf{C}_{\text{TX},k}^\trasp \right \|^2
\ee
from which one gets
\be
g_{\text{TX},k}^2 =\frac{\trace(\bm{\Phi}_{\text{TX},k} \, \mathbf{C}_{\text{TX},k}^*
\, \bm{\Omega}^* \, \bm{\Omega}^\herm \,
\mathbf{C}_{\text{TX},k}^\herm)}
{\trace(\mathbf{C}_{\text{TX},k} \, \bm{\Omega} \, \bm{\Omega}^\trasp \,
\mathbf{C}_{\text{TX},k}^\trasp \mathbf{C}_{\text{TX},k}^*
\, \bm{\Omega}^* \, \bm{\Omega}^\herm \,
\mathbf{C}_{\text{TX},k}^\herm)}
\label{eq:g}
\ee
for $k \in \{1,2,\ldots, K_\text{TX}\}$ and
$\text{TX} \in \{\text{U},\text{J}\}$.
Since \eqref{eq:g} provides the square of
the complex gain $g_{\text{TX},k}$,
there
is a $\pi$-radian phase ambiguity associated with the
recovery of $\measuredangle g_{\text{TX},k}$.
Such a sign ambiguity can be removed
by building all the $2^{\Ka+\Ks}$ possible
matrices $\Htilde[n] $ in
\eqref{eq:ytilde} and, then, choosing
the one that yields the best
pre-detection performance.
The whole estimation process is summarized in
Algorithm~\ref{table:tab_3}
reported at the top of this page.

In practice, an estimate $\widehat{\tau}_{\text{TX},k}$ of the
delay $\tau_{\text{TX},k}$ is derived by searching for the peak of
the finite-sample version $\widehat{\mathcal{I}}_{\text{TX},k}(\beta)$ of the cost function
$\mathcal{I}_{\text{TX},k}(\beta)$ defined in \eqref{eq:I}, where
$\widehat{\mathcal{I}}_{\text{TX},k}(\beta)$ is obtained by replacing
$\Rb_{\overline{\yb}\overline{\yb}^*}^{2 \, \nu_{\text{TX},k}}[r]$
in \eqref{eq:Phi-1}--\eqref{eq:Phi+1},
for $r \in \{-1,0,1\}$,
with their corresponding estimates \eqref{eq:Rciclest},
evaluated at the
estimated cycle frequency
$\alpha= 2\, \widehat{\nu}_{\text{TX},k}$.
Additionally, the estimate $\widehat{\nu}_{\text{TX},k}$ of the
$k$th Doppler shift associated with the TX-to-CU channel -- which
has been previously achieved as explained
in Subsection~\ref{sec:Doppler} -- has to be used to replace
$\mathbf{D}_{\text{TX},k}$ in \eqref{eq:Phi-1}--\eqref{eq:Phi+1}
with its estimate
$\widehat{\mathbf{D}}_{\text{TX},k} \eqdef \diag(1, e^{j \, \frac{2 \pi}{P} \widehat{\nu}_{\text{TX},k}}, \ldots, e^{j \, \frac{2 \pi}{P} \widehat{\nu}_{\text{TX},k} (P-1)})$.
So doing, one constructs
the matrices
\barr
\widehat{\bm{\Phi}}_{\text{TX},k}[-1] & =e^{-j \, 2 \pi \widehat{\nu}_{\text{TX},k}}  \,
\widehat{\mathbf{D}}_{\text{TX},k}^* \, \widehat{\Rb}_{\overline{\yb}\overline{\yb}^*}^{2 \,
\widehat{\nu}_{\text{TX},k}}[-1] \, \widehat{\mathbf{D}}_{\text{TX},k}^*
\nonumber \\
\widehat{\bm{\Phi}}_{\text{TX},k}[0]  & =
\widehat{\mathbf{D}}_{\text{TX},k}^* \, \widehat{\Rb}_{\overline{\yb}\overline{\yb}^*}^{2 \, \widehat{\nu}_{\text{TX},k}}[0]
\, \widehat{\mathbf{D}}_{\text{TX},k}^*
\\
\widehat{\bm{\Phi}}_{\text{TX},k}[1] & =
e^{j \, 2 \pi \widehat{\nu}_{\text{TX},k}} \, \widehat{\mathbf{D}}_{\text{TX},k}^* \, \widehat{\Rb}_{\overline{\yb}\overline{\yb}^*}^{2 \, \widehat{\nu}_{\text{TX},k}}[1] \,
\widehat{\mathbf{D}}_{\text{TX},k}^*
\\
\widehat{\bm{\Phi}}_{\text{TX},k} & =\widehat{\bm{\Phi}}_{\text{TX},k}[-1] +
\widehat{\bm{\Phi}}_{\text{TX},k}[0]+ \widehat{\bm{\Phi}}_{\text{TX},k}[1] \: .
\earr
After calculating $\widehat{\tau}_{\text{TX},k}$, the corresponding estimate
$\widehat{g}_{\text{TX},k}$ of the channel gain $g_{\text{TX},k}$ is obtained
from \eqref{eq:g} by replacing $\bm{\Phi}_{\text{TX},k}$ and
$\mathbf{C}_{\text{TX},k}$ with their corresponding estimates
$\widehat{\bm{\Phi}}_{\text{TX},k}$ and
$\widehat{\mathbf{C}}_{\text{TX},k}=\bm{W}_P \,  \diag(\widehat{\vb}_{\text{TX},k}) \,
\bm{W}_P^\herm$, respectively,
where the $p$th entry of
$\widehat{\vb}_{\text{TX},k}$
is given by $\{\widehat{\vb}_{\text{TX},k}\}_p= \widehat{\Psi}_{\text{TX},k}[p] \,
e^{-j \frac{2 \pi}{T} \widehat{\tau}_{\text{TX},k} \, p}$, for $p \in \{0,1,\ldots, P-1\}$,
and $\widehat{\Psi}_{\text{TX},k}[p]$ is obtained from \eqref{eq:59}
in the Appendix  by replacing $\chi_{\text{TX},k}$ with $\widehat{\chi}_{\text{TX},k}$.

Study of the asymptotic
behavior of the proposed estimators 
of channel gains and time delays
is a challenging mathematical problem. The performances of these estimators
are assessed by simulations in Section~\ref{sec:simul}.

\begin{table}[t]
\caption{MSE of UAV channel parameters (high-speed jammer).}
\label{tab:Uav-h}
\centering{}%
\begin{tabular}{cccc}
\hline
\noalign{\vskip\doublerulesep}
\textbf{MSE} (dB) & SJR = $-3$ dB & SJR = $0$ dB & SJR = $3$ dB \tabularnewline[\doublerulesep]
\hline
\noalign{\vskip\doublerulesep}
\hline
\noalign{\vskip\doublerulesep}
Channel gain  &
$-33.4944$ &  $-37.1155$ & $-39.7152$
\tabularnewline[\doublerulesep]
\hline
\noalign{\vskip\doublerulesep}
Doppler shift
&
$-63.8131$ & $-63.8131$ & $-63.8859$
\tabularnewline[\doublerulesep]
\hline
\noalign{\vskip\doublerulesep}
Time delay
&
$-39.7772$ & $-43.4395$ & $-46.1916$
\tabularnewline[\doublerulesep]
\hline
\hline
\end{tabular}
\end{table}
\begin{table}[t]
\caption{MSE of jammer channel parameters (high-speed jammer).}
\label{tab:Jammer-h}
\centering{}%
\begin{tabular}{cccc}
\hline
\noalign{\vskip\doublerulesep}
\textbf{MSE} (dB) & SJR = $-3$ dB & SJR = $0$ dB & SJR = $3$ dB \tabularnewline[\doublerulesep]
\hline
\noalign{\vskip\doublerulesep}
\hline
\noalign{\vskip\doublerulesep}
Channel gain  &
$-41.5433$ & $-38.4094$ & $-34.2107$
\tabularnewline[\doublerulesep]
\hline
\noalign{\vskip\doublerulesep}
Doppler shift
&
$-70.0077$ & $-70.0077$ & $-69.9169$
\tabularnewline[\doublerulesep]
\hline
\noalign{\vskip\doublerulesep}
Time delay
&
$-46.6658$ & $-43.1931$ & $-38.3779$
\tabularnewline[\doublerulesep]
\hline
\hline
\end{tabular}
\end{table}
\begin{figure*}[!t]
\begin{minipage}[b]{9cm}
\centering
\includegraphics[width=\linewidth]{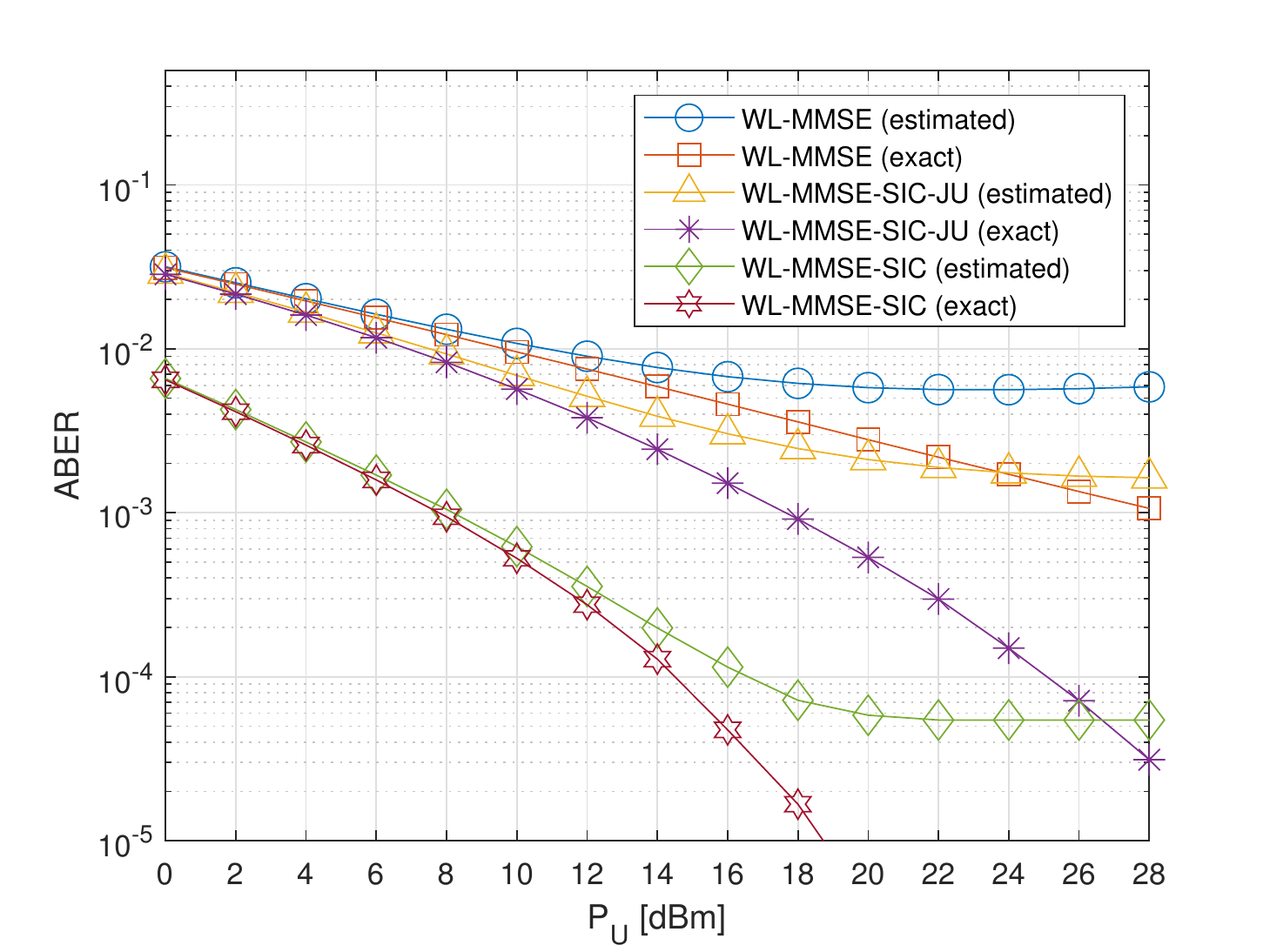}
\end{minipage}
\begin{minipage}[b]{8cm}
\centering
\includegraphics[width=0.99\linewidth]{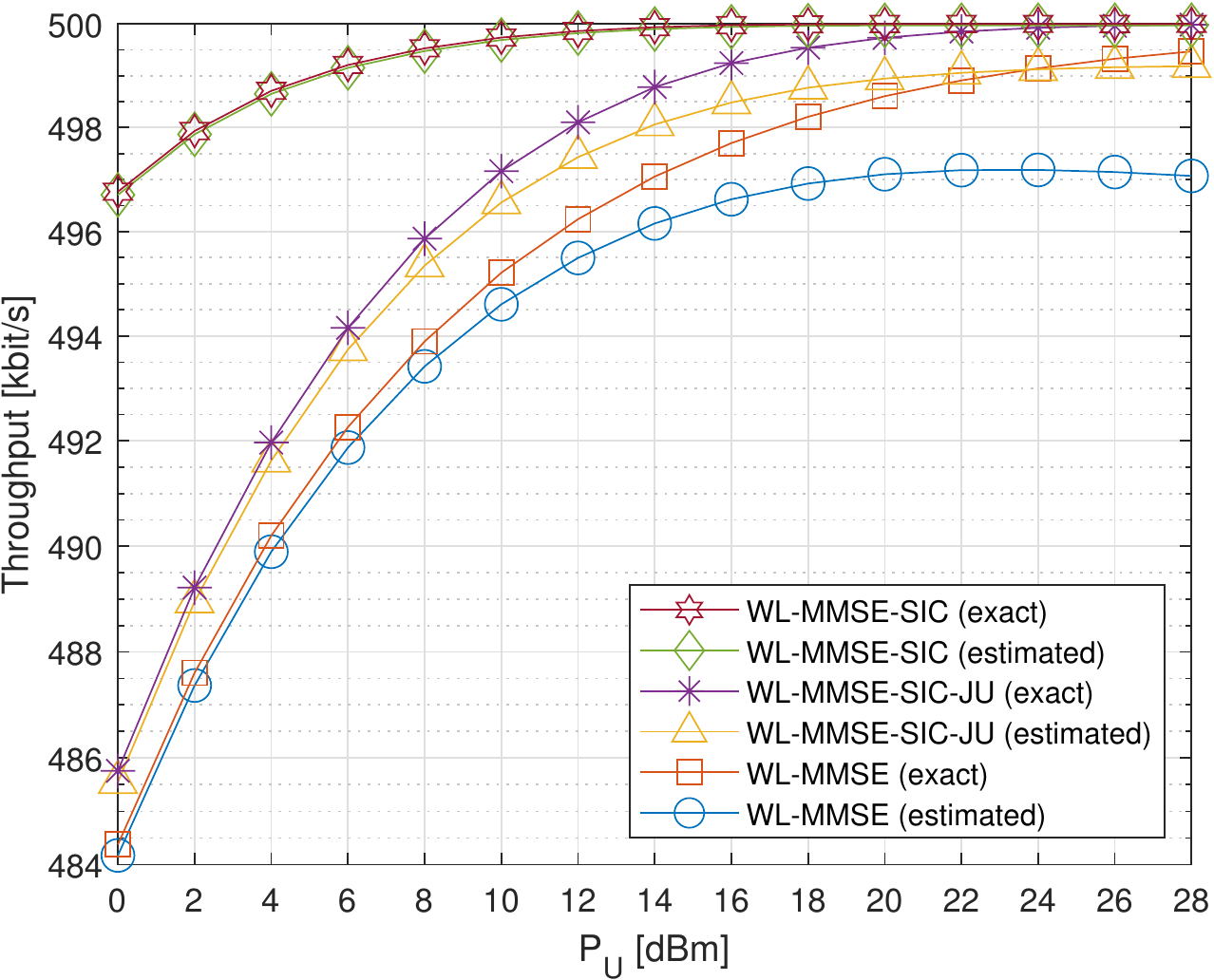}
\end{minipage}
\caption{ABER (left) and throughput  (right) versus SNR (SJR = $-3$ dB, high-speed jammer).
}
\label{fig:fig_2}
\end{figure*}
\begin{figure*}[!t]
\begin{minipage}[b]{9cm}
\centering
\includegraphics[width=\linewidth]{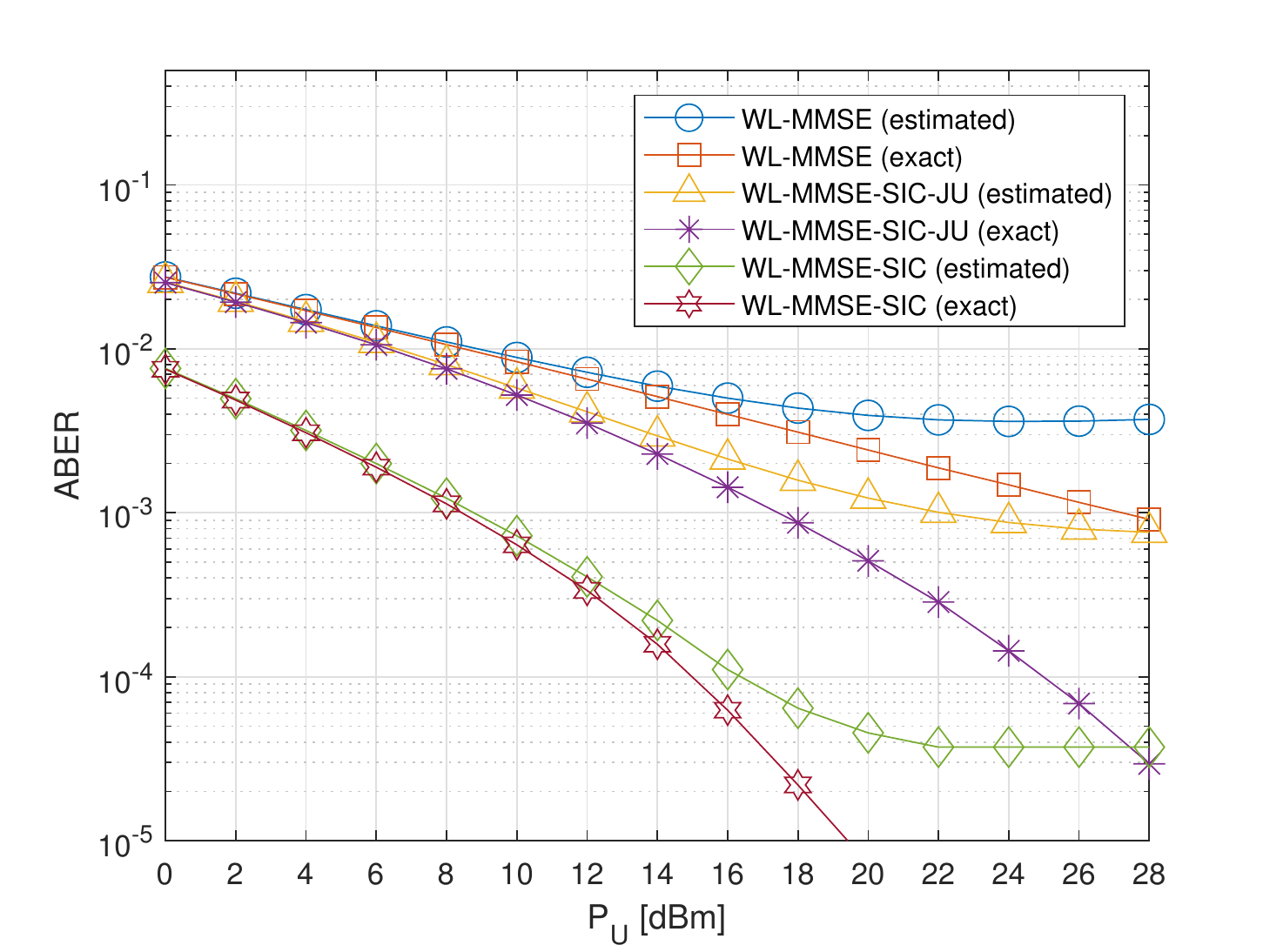}
\end{minipage}
\begin{minipage}[b]{8cm}
\centering
\includegraphics[width=0.99\linewidth]{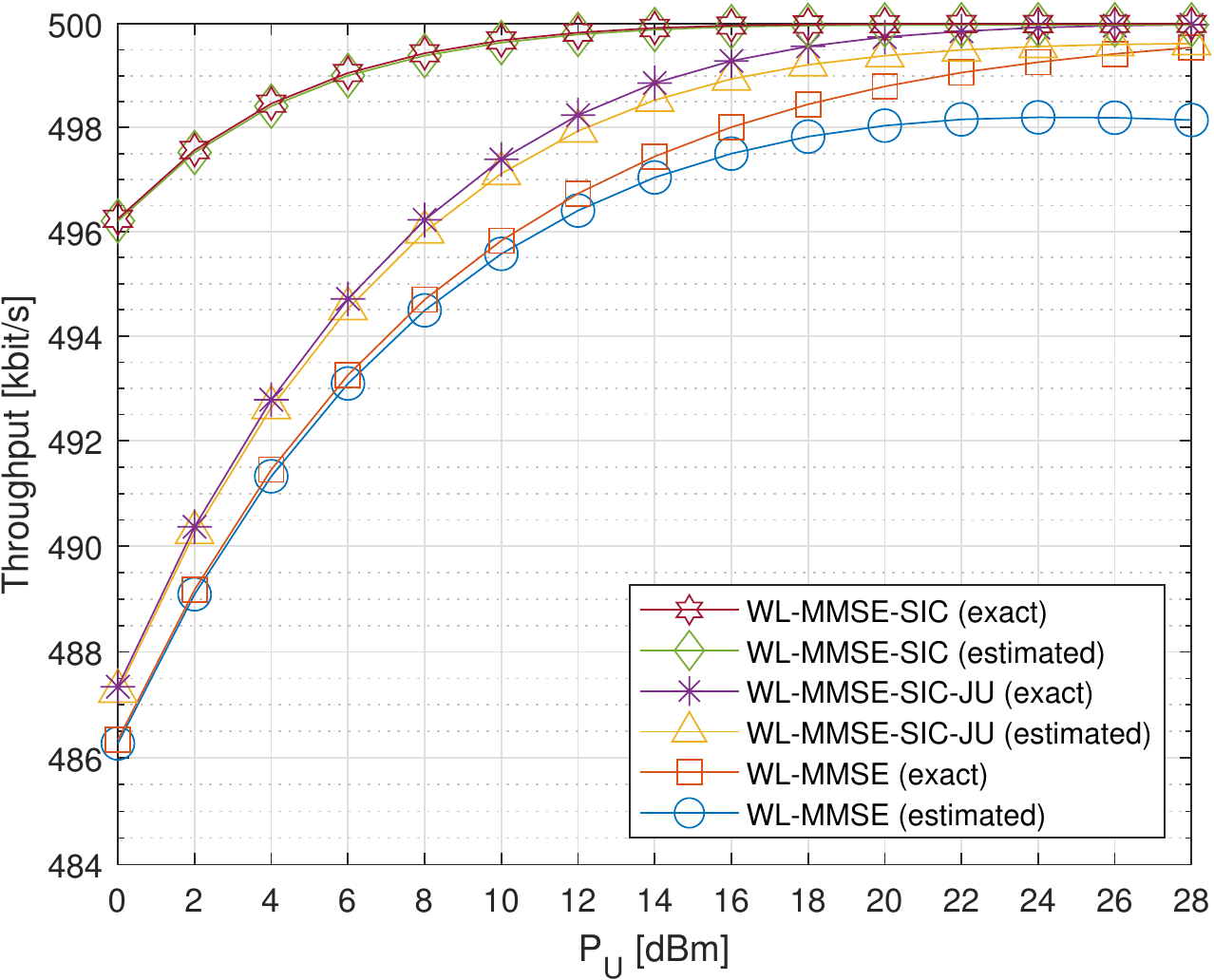}
\end{minipage}
\caption{ABER (left) and throughput  (right) versus SNR (SJR = $0$ dB, high-speed jammer).
}
\label{fig:fig_3}
\end{figure*}
\begin{figure*}[!t]
\begin{minipage}[b]{9cm}
\centering
\includegraphics[width=\linewidth]{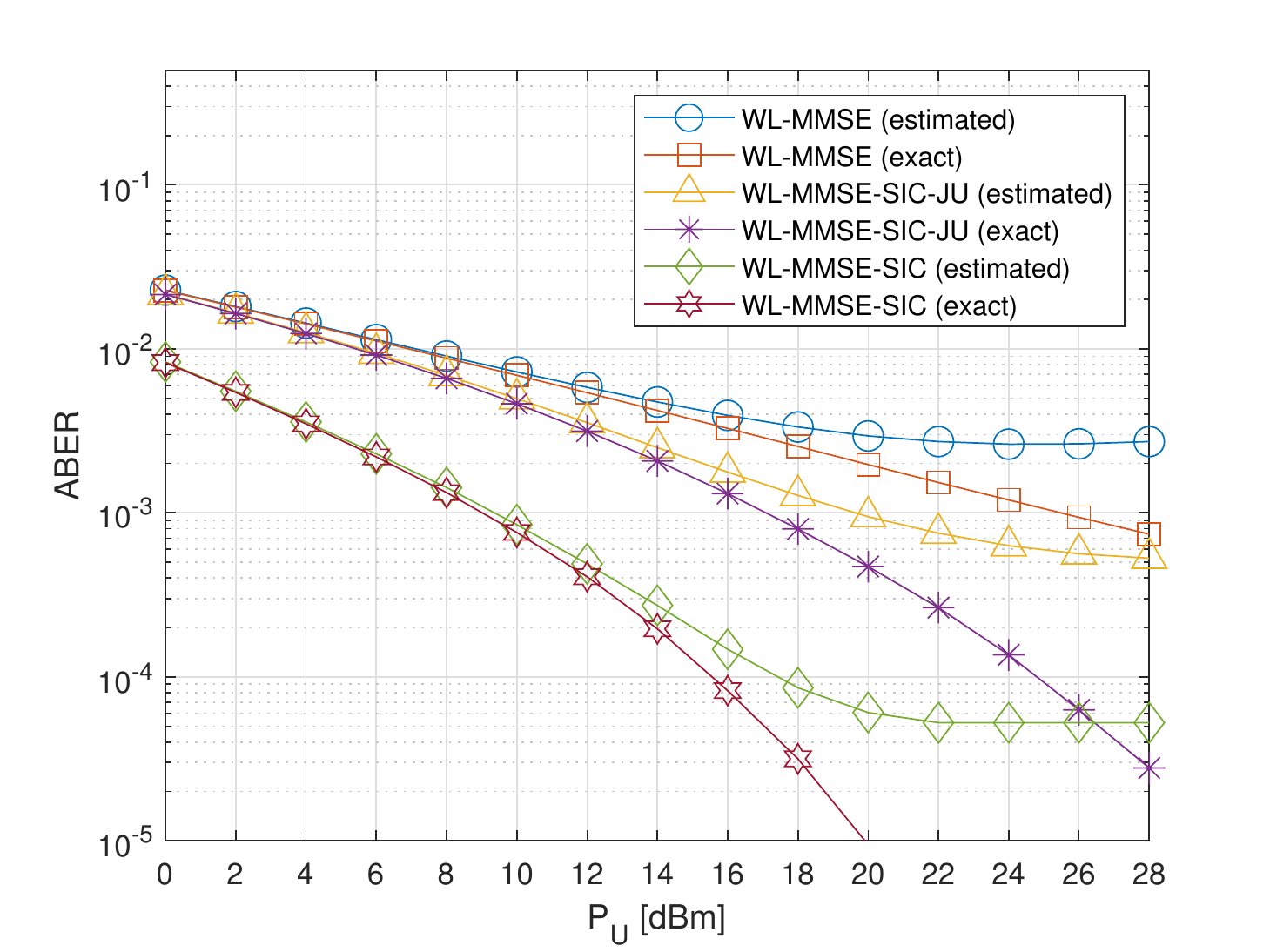}
\end{minipage}
\begin{minipage}[b]{8cm}
\centering
\includegraphics[width=0.99\linewidth]{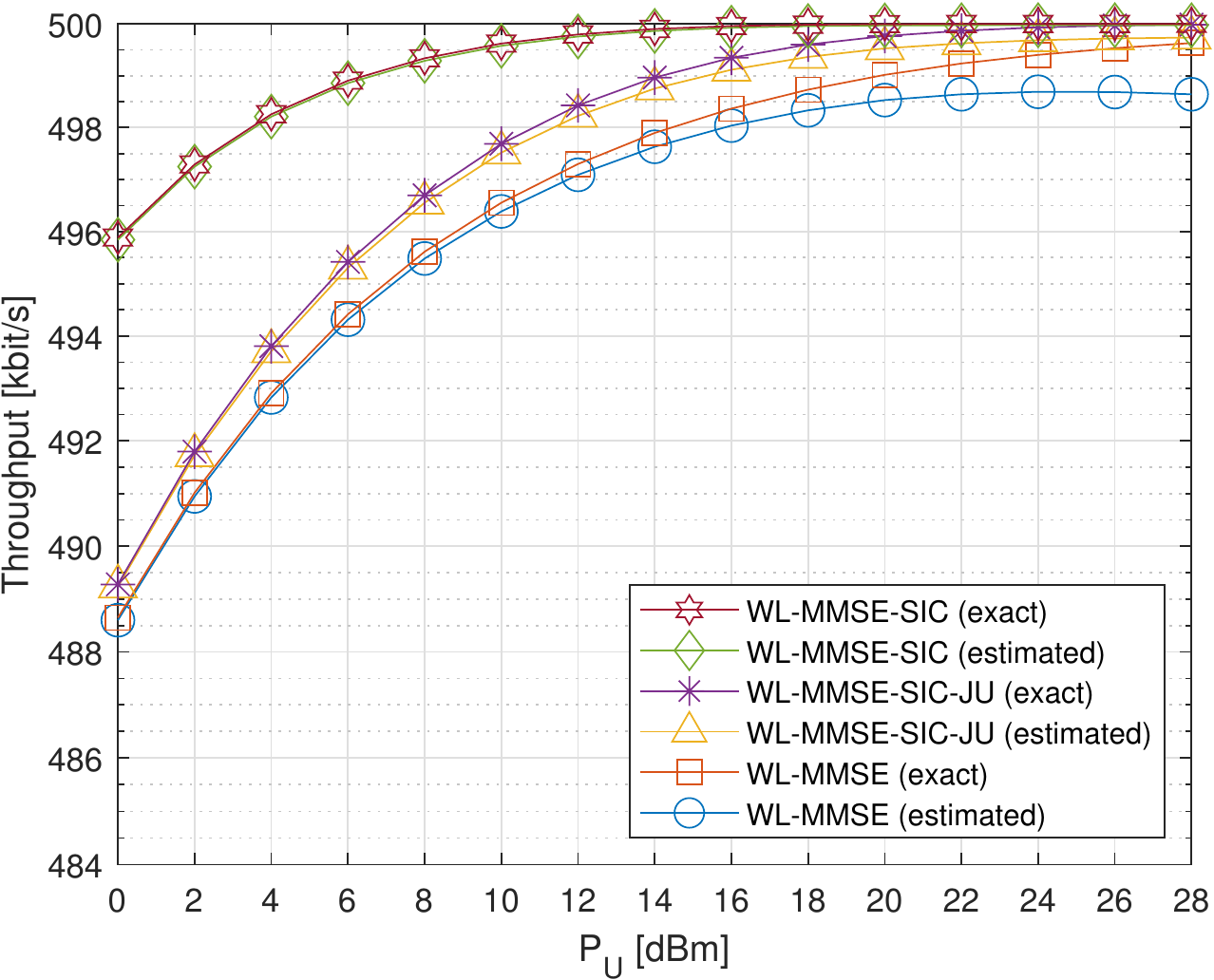}
\end{minipage}
\caption{ABER (left) and throughput  (right) versus SNR (SJR = $3$ dB, high-speed jammer).
}
\label{fig:fig_4}
\end{figure*}

\begin{table}[t]
\caption{MSE of UAV channel parameters (low-speed jammer).}
\label{tab:Uav-l}
\centering{}
\begin{tabular}{cccc}
\hline
\noalign{\vskip\doublerulesep}
\textbf{MSE} (dB) & SJR = $-3$ dB & SJR = $0$ dB & SJR = $3$ dB \tabularnewline[\doublerulesep]
\hline
\noalign{\vskip\doublerulesep}
\hline
\noalign{\vskip\doublerulesep}
Channel gain  &
$-35.4503$ &  $-39.4990$ & $-42.3258$
\tabularnewline[\doublerulesep]
\hline
\noalign{\vskip\doublerulesep}
Doppler shift
&
$-63.4707$ & $-63.4707$ & $-63.4707$
\tabularnewline[\doublerulesep]
\hline
\noalign{\vskip\doublerulesep}
Time delay
&
$-38.7822$ & $-43.5352$ & $-46.5625$
\tabularnewline[\doublerulesep]
\hline
\hline
\end{tabular}
\end{table}
\begin{table}[t]
\caption{MSE of jammer channel parameters (low-speed jammer).}
\label{tab:Jammer-l}
\centering{}
\begin{tabular}{cccc}
\hline
\noalign{\vskip\doublerulesep}
\textbf{MSE} (dB) & SJR = $-3$ dB & SJR = $0$ dB & SJR = $3$ dB \tabularnewline[\doublerulesep]
\hline
\noalign{\vskip\doublerulesep}
\hline
\noalign{\vskip\doublerulesep}
Channel gain  &
$-41.9432$ & $-38.8324$ & $-34.4111$
\tabularnewline[\doublerulesep]
\hline
\noalign{\vskip\doublerulesep}
Doppler shift
&
$-57.1067$ & $-57.2068$ & $-57.2068$
\tabularnewline[\doublerulesep]
\hline
\noalign{\vskip\doublerulesep}
Time delay
&
$-45.9246$ & $-43.2448$ & $-39.4336$
\tabularnewline[\doublerulesep]
\hline
\hline
\end{tabular}
\end{table}
\begin{figure*}[!t]
\begin{minipage}[b]{9cm}
\centering
\includegraphics[width=\linewidth]{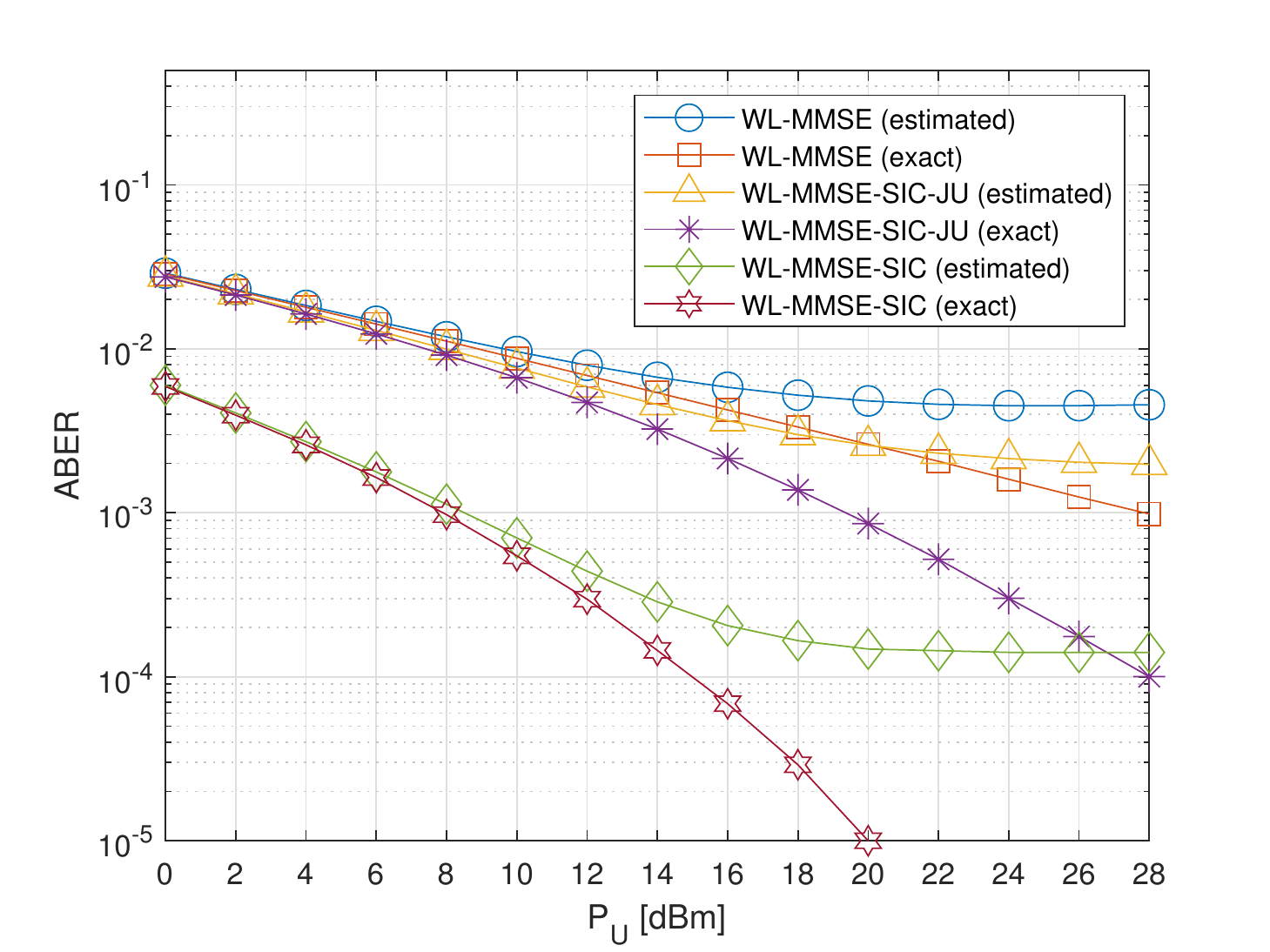}
\end{minipage}
\begin{minipage}[b]{8cm}
\centering
\includegraphics[width=0.99\linewidth]{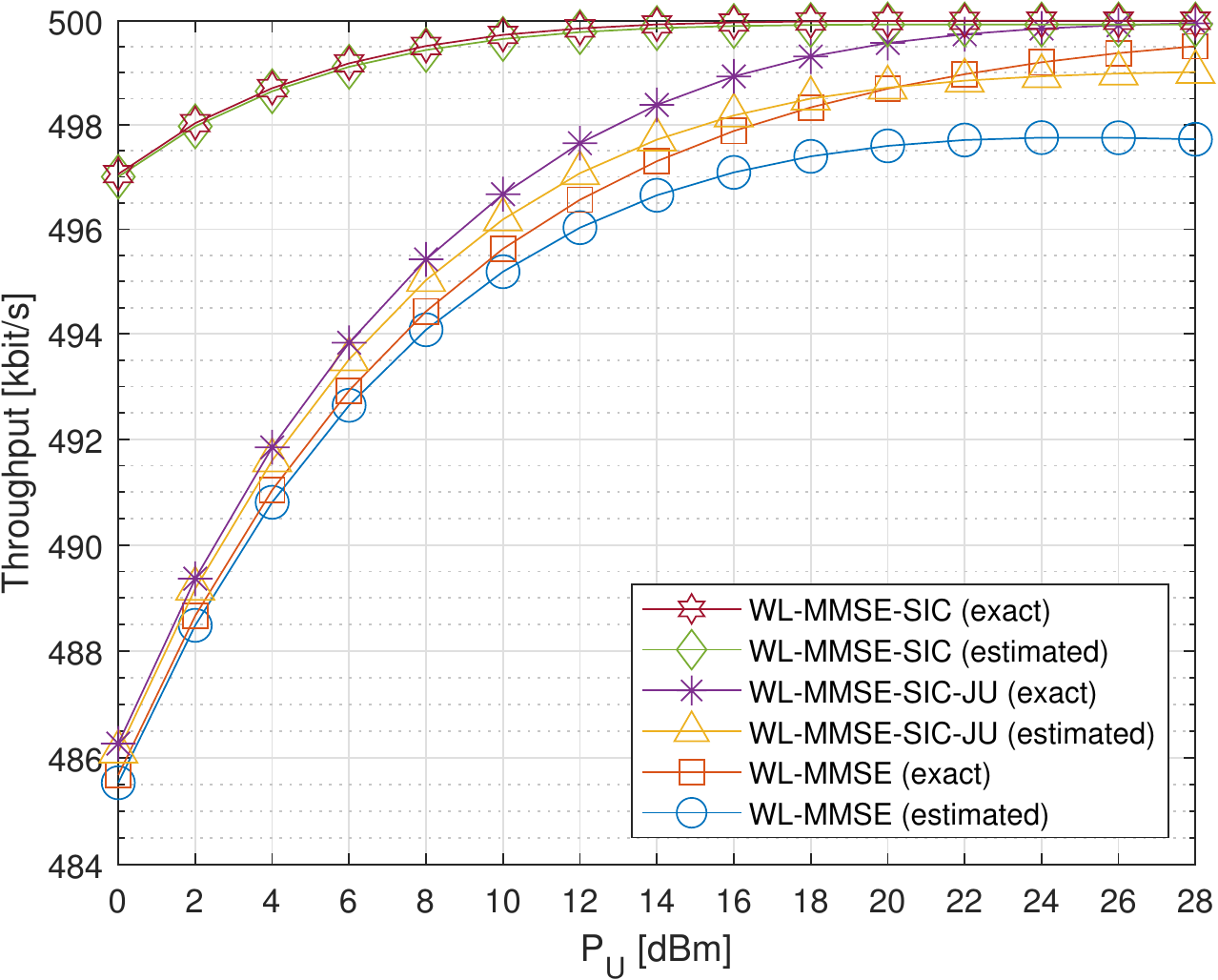}
\end{minipage}
\caption{ABER (left) and throughput  (right) versus SNR (SJR = $-3$ dB, low-speed jammer).
}
\label{fig:fig_5}
\end{figure*}
\begin{figure*}[!t]
\begin{minipage}[b]{9cm}
\centering
\includegraphics[width=\linewidth]{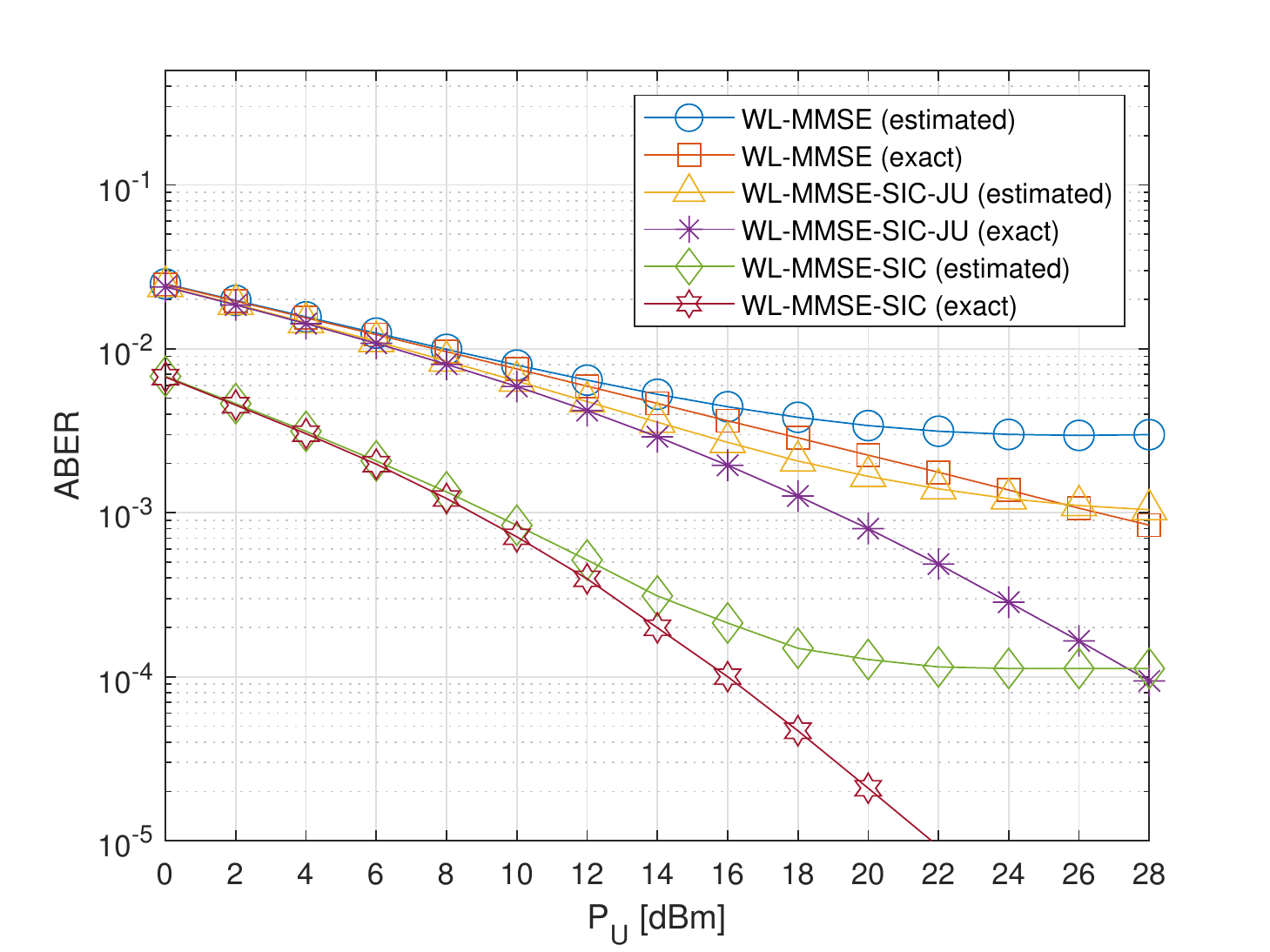}
\end{minipage}
\begin{minipage}[b]{8cm}
\centering
\includegraphics[width=0.99\linewidth]{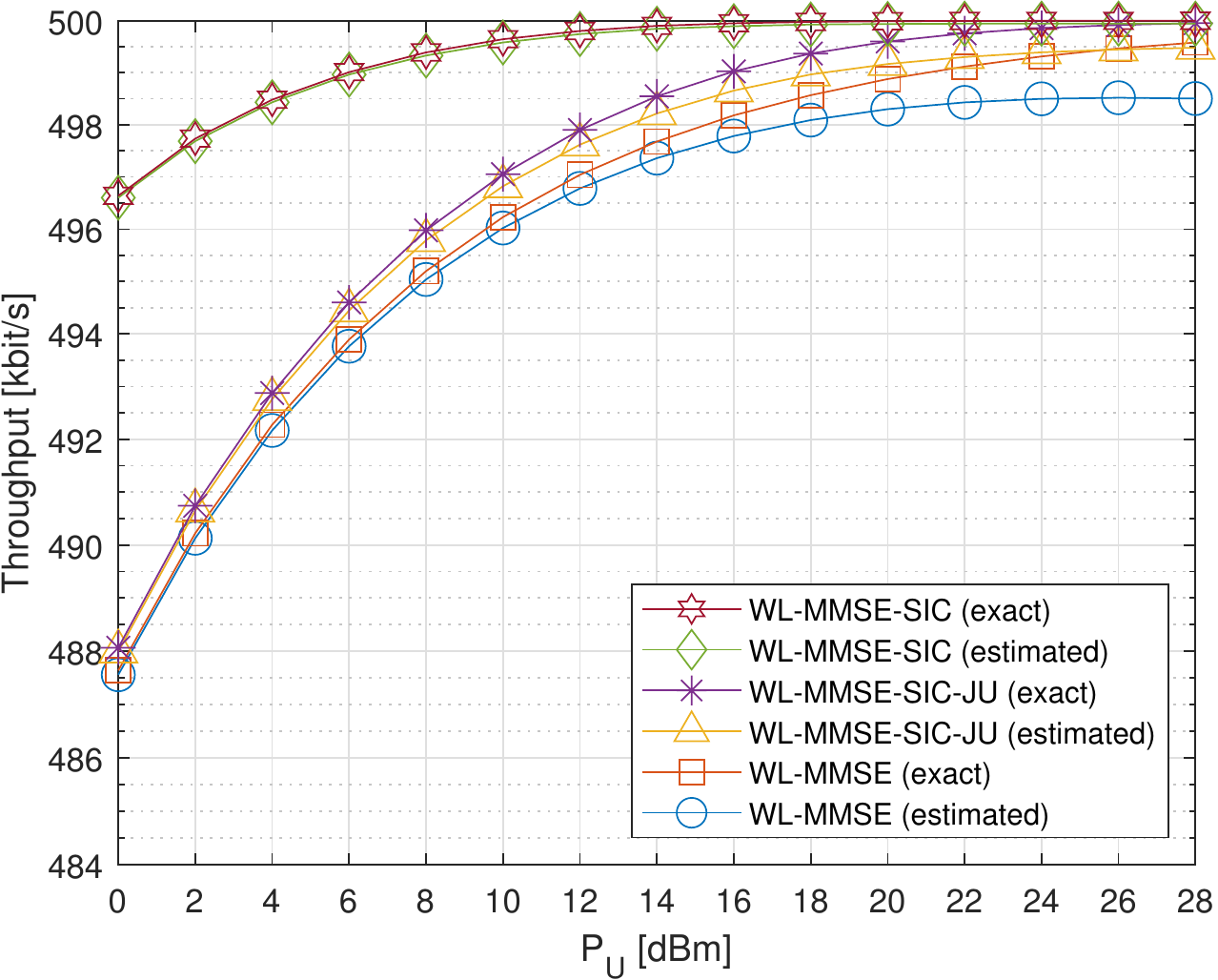}
\end{minipage}
\caption{ABER (left) and throughput  (right) versus SNR (SJR = $0$ dB, low-speed jammer).
}
\label{fig:fig_6}
\end{figure*}
\begin{figure*}[!t]
\begin{minipage}[b]{9cm}
\centering
\includegraphics[width=\linewidth]{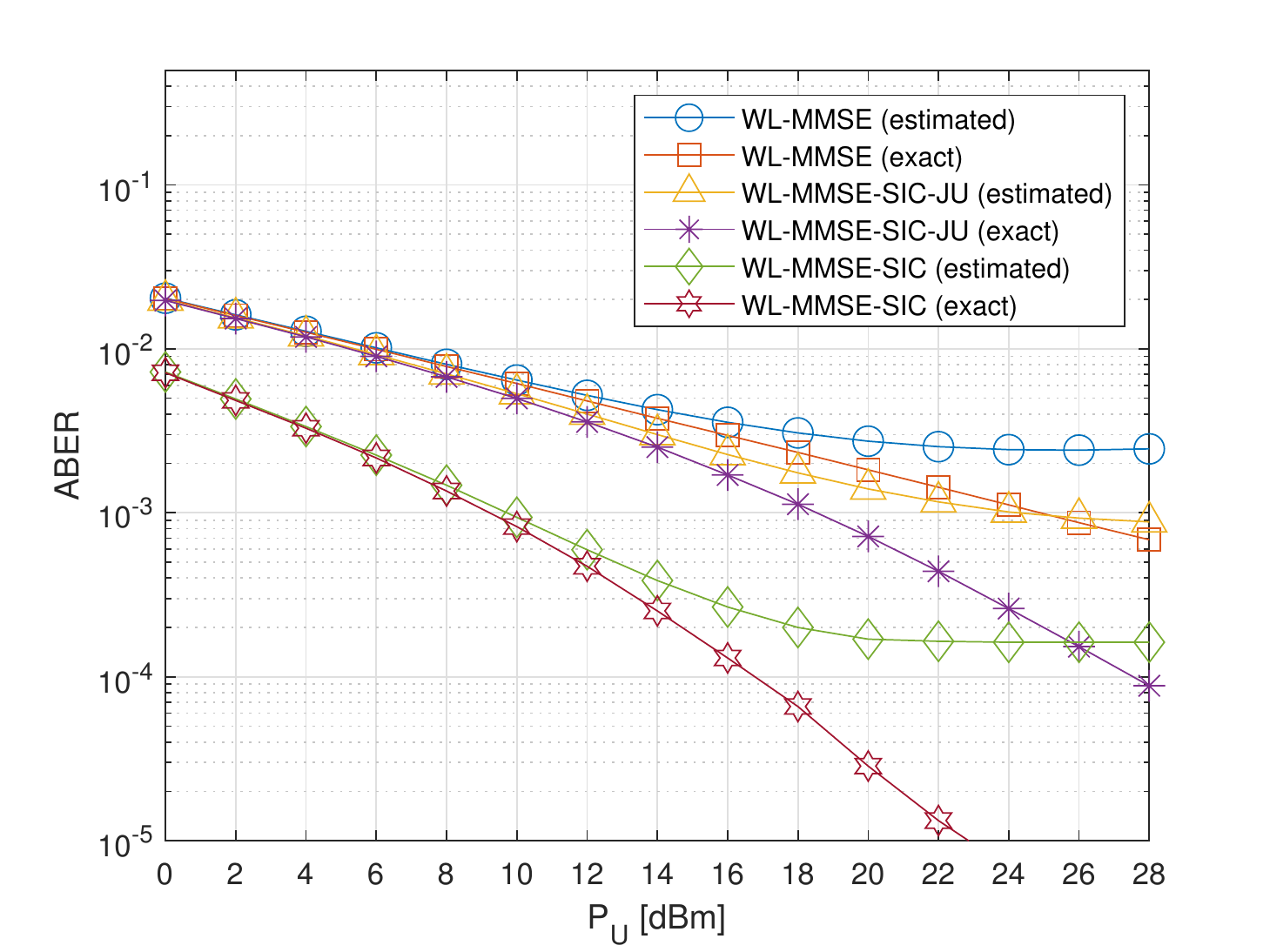}
\end{minipage}
\begin{minipage}[b]{8cm}
\centering
\includegraphics[width=0.99\linewidth]{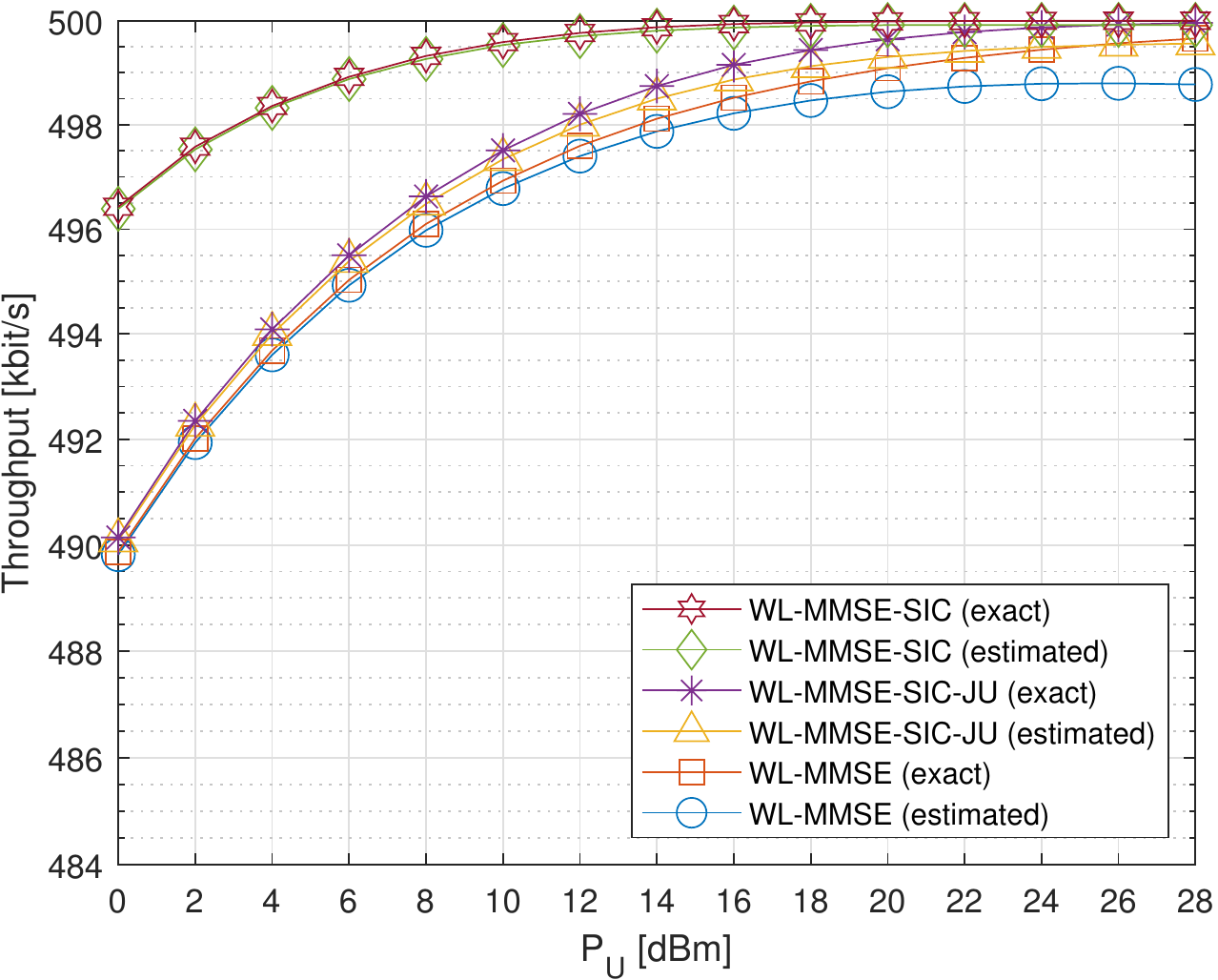}
\end{minipage}
\caption{ABER (left) and throughput  (right) versus SNR (SJR = $3$ dB, low-speed jammer).
}
\label{fig:fig_7}
\end{figure*}

\section{Numerical results}
\label{sec:simul}

In this Section, we provide simulation results aimed at evaluating the
performance of the considered anti-jamming detector (see Algorithm~\ref{table:tab_1}),
equipped with the proposed channel estimators (see Algorithms~\ref{table:tab_2} and \ref{table:tab_3}).
To this aim, we consider the following simulation setting.
The UAV and the mobile jammer employ OFDM modulation with $M=16$ subcarriers,
CP of length $\Lcp=4$, sampling rate $1/T_\text{c}=625$ kHz, and
binary-phase-shift-keying (BPSK) signaling.
The carrier frequency is set to $f_0=27$ GHz.
The number of paths of the UAV-to-CU and jammer-to-CU
links are fixed to  $K_{\text{U}}=K_{\text{J}}=2$.
For $k \in \{1, 2\}$ and $\text{TX} \in \{\text{U},\text{J}\}$, the gains
$g_{\text{TX},k}$ of both the UAV and jammer channels
are generated as statistically independent Rayleigh random variables (RVs), with
\be
\sigma_\text{TX}^2 \eqdef \Es(|g_{\text{TX},k}|^2) = 2 \, \pot_\text{TX} \left(\frac{\lambda_0}
{4 \pi  d_\text{TX}}\right)^2
\ee
where $\lambda_0=c_0/f_0$ is the wavelength,
$c_0 = 3 \cdot 10^8$ m/s
denotes the light speed, and $d_\text{TX}$ is the distance
between TX and the CU,
whereas the delays $\tau_{\text{TX},k}$ are randomly generated according to the
one-sided exponentially decreasing delay power spectrum \cite{Haas}, i.e.,
\be
\tau_{\text{TX},k} = -\tau_{\text{slope}} \, \text{ln} [1-u_k  (1-e^{-\Delta_{\text{TX}}/\tau_{\text{slope}}})]
\ee
where the maximum delay is $\Delta_{\text{U}}= \Delta_{\text{J}}=3 \, T_\text{c}$, the slope-time is
$\tau_{\text{slop}} = 2 \, T_\text{c}$, and $u_k$ are independent 
RVs uniformly distributed in $[0,1]$.
We set $d_\text{U}=d_\text{J}=100$ m and
the noise variance $\sigma^2_w$ is equal 
to $-113$ dBm.

We assume that the flight speed of the UAV 
is $v_{\text{U}} = 10$ m/s, while we consider
two mobility scenarios for the jammer:

\begin{enumerate}

\itemsep=0mm

\item
{\em high-speed jammer}: the jammer moves at $v_{\text{J}} = 20$ m/s;

\item
{\em low-speed jammer}: the jammer moves at $v_{\text{J}} = 5$ m/s.

\end{enumerate}
The Doppler frequencies $f_{\text{TX},k}$ are customarily generated as
$f_{\text{TX},k} = D_{\text{TX}} \, \text{cos}(\vartheta_{\text{TX},k})$, where
$D_{\text{TX}} \eqdef (v_{\text{TX}}/c_0) f_0$ and $\vartheta_{\text{TX},k}$ are independent RVs, uniformly
distributed in $[0,2 \pi)$.
During the downloading phase, the UAV transmits $5$ MB 
of collected data.

In each Monte Carlo run, a new set of symbols, noise and channel parameters for both the UAV and jammer are randomly generated.
As a performance measure, we report the average bit error rate (ABER) at the CU regarding
the UAV transmission and the corresponding throughput,\footnote{For the considered simulation setting, the maximum achievable throughput is given by $(1/T_\text{c})\, (M/P)=625 \cdot 0.8=500$
kbit/s, due to the redundancy arising from the CP insertion at the OFDM transmitter.}
as well as
the arithmetic mean of the MSEs of the Doppler shifts
$f_{\text{TX},k}$ (normalized by $D^2_\text{TX}$),  the time delays
$\tau_{\text{TX},k}$ (normalized by $\Delta^2_\text{TX}$),
and channel gains  $g_{\text{TX},k}$ (normalized by $\sigma^2_\text{TX}$),
for
$k \in \{1,2\}$ and
$\text{TX} \in \{\text{U},\text{J}\}$.
Three different values 
of $\text{SJR} \in \{-3,0,3\}$ dB are considered,
with $\text{SJR} \eqdef \sigma_\text{U}^2/\sigma_\text{J}^2$,
representing the situations in which the 
transmit power of the UAV is half, equal to, or double of 
that of the jammer, respectively.

In addition to the proposed WL symbol detection procedure
with anti-jamming  capabilities based on
serial interference cancellation (SIC), we also
evaluated the performance of the simpler
WL-MMSE detector.
We also implemented the WL-MMSE-SIC
{\em jamming-unaware (JU)} detector, which
recovers the UAV symbol vector $\bm{s}_\text{U}[n]$
-- rather than both $\bm{s}_\text{U}[n]$ and $\bm{s}_\text{J}[n]$ --
and involves only the knowledge of $\Ha[n]$,
thereby treating
the jammer as a disturbance.
The ABER curves of all the considered detection
techniques are reported in the case of perfect knowledge of
$\Ha[n]$ and $\Hs[n]$,
referred to as ``exact'', as well as when
the channel state information
is built by using the
estimates of the relevant channel parameters
(i.e., Doppler shifts, time delays, and channel gains)
obtained from the received data,
referred to as ``estimated''.

\subsection{Scenario 1: High-speed jammer}
The performance of the proposed finite-sample versions of the
estimators outlined in  Algorithms~\ref{table:tab_2} and  \ref{table:tab_3}
is reported in Tables~\ref{tab:Uav-h} and \ref{tab:Jammer-h} for the
UAV and jammer transmissions, respectively, when $\pot_\text{U}=10$ dBm
and $v_{\text{J}} = 20$ m/s.
It is apparent that the proposed estimators exhibit very satisfactory
MSE performance for all the considered SJR values.
Remarkably, the estimation accuracy of the Doppler shifts is almost
insensitive to the SJR. It is also evident that, compared with the
UAV case,  the channel parameters of the jammer are estimated with a better accuracy.
This is due to the fact that a larger value of the transmitter speed enables
a better estimate of its cycle frequencies, which in its turn has a positive effect on
the estimation process of the time delays and channel gains, whose
estimators require prior acquisition of the Doppler shifts.
Results not reported here show that the estimation accuracy
is the same for different values of $\pot_\text{U}$.

Figs.~\ref{fig:fig_2}, \ref{fig:fig_3}, and \ref{fig:fig_4} depict the ABER (left-side plot)
and throughput (right-side plot) performance
of the WL-MMSE  detector and the proposed WL-MMSE-SIC one, whose
synthesis is detailed in Algorithm~\ref{table:tab_1}, as well as
the performance of the
WL-MMSE-SIC-JU receiver.
As confirmed by our theoretical prediction, the system performance when the
WL-MMSE detector is employed at the CU is seriously degraded by the residual
ICI and jamming contributions at its output. In particular, the data-estimated version
of the WL-MMSE detector exhibits a high ABER floor well above $10^{-3}$,
which in its turn leads to a corresponding throughput floor.
Compared with the WL-MMSE detector,
the WL-MMSE-SIC-JU receiver allows one to significantly improve the system performance
in ideal conditions (i.e., exact knowledge of $\Ha[n]$ and the correlation matrix of
$\z_i$, for $i \in \{0,1,\ldots, 2M-1\}$), but pays an acceptable performance degradation
when implemented from data, due to the presence of a high-power jammer.
On the other hand, the proposed WL-MMSE-SIC detector is able to effectively
remove the residual ICI and jamming impairments after the WL-MMSE pre-detection,
by ensuring performances that are almost unaffected by the considered SJR values
in both the exact and estimated versions. Specifically, the data-estimated
version of the WL-MMSE-SIC shows a BER floor only below $10^{-4}$,
due to the finite-sample estimation of the relevant channel parameters,
which does not prevent to achieve the maximum throughput.

\subsection{Scenario 2: Low-speed jammer}

The numerical experiment of the previous Section is repeated for the case of $v_{\text{J}} = 5$ m/s
and the corresponding results are reported in Tables~\ref{tab:Uav-l} and \ref{tab:Jammer-l},
as well as in Figs.~\ref{fig:fig_5}, \ref{fig:fig_6}, and \ref{fig:fig_7}.
It is seen that, as intuitively expected, the only noticeable effect of  a lower speed of the jammer is
a reduction of the estimation accuracy of its Doppler shifts, 
which however has a negligible impact on ABER and throughput performance of the UAV transmission.

\section{Conclusions and directions for future work}
\label{sec:concl}

In this paper, we have investigated anti-jamming
communications in UAV-aided WSNs operating over
doubly-selective channels
for smart city applications. We have focused on the
downloading phase, when the UAV is subject to a jamming attack when it is transmitting
to a remote CU the data collected through the city.

The following countermeasures have been proposed:
\begin{enumerate}[(i)]
\item
Joint detection of the UAV and jammer symbols
is performed to address the case of strong ICI and
jamming power, in conjunction with serial
cancellation of the residual ICI and jamming
contribution based on post-sorting of the
detector output.
\item
The Doppler shifts, time delays and channel gains
required to implement the improved detection
strategy are blindly acquired by exploiting
the ACS properties of the received signals
through algorithms
that exploit the detailed
structure of multicarrier modulation format.
\end{enumerate}

In summary, this study demonstrates that, if sophisticated
reception strategies and channel estimators are employed during jamming attack,
it is still possible to ensure a satisfactory performance 
in terms of ABER, throughput,
and MSE of the estimated channel parameters, which is only
slightly affected by the transmitted power of the jammer.
It is worth noting that the proposed countermeasures 
can be used  outside the application
scenario of smart cities, whenever a large amount of data is available
and latency constraints are relaxed.
Finally, we have studied the performance of the proposed
estimators by means of numerical simulations.
In this respect,
a first interesting research subject consists of
investigating the consistency
and asymptotic distribution of the
developed estimators from a
theoretical viewpoint.
To further corroborate the effectiveness of the proposed
anti-jamming countermeasures,  an additional research issue
is to deploy a physical testbed by programming
general purpose hardware according to the software-defined radio paradigm.

\appendices

\section{Proof of Proposition~\ref{prop:1}}
\label{app:1}
Lat us take into account two different linear mappings:
the former one is given by
\be
\mathcal{L}_1:
\boldsymbol{\nu}_\text{ord} \in [-1/2,1/2)^{K}
\longrightarrow \boldsymbol{\alpha}_\text{ord} \in [-1/2,1/2)^{L_\text{a}}
\label{eq:L1}
\ee
where we remember that $\boldsymbol{\nu}_\text{ord} =(\boldsymbol{\nu}_\text{U}^\trasp,
\boldsymbol{\nu}_\text{J}^\trasp)^\trasp \in [-1/2,1/2)^{K}$,
with $K \eqdef K_\text{U}+K_\text{J}$, and
\begin{multline}
\boldsymbol{\alpha}_\text{ord} \eqdef (\alpha_{\text{U}^{(*)},1,1}, \ldots,
\alpha_{\text{U}^{(*)},\Ka,\Ka},
\alpha_{\text{J}^{(*)},1,1}, \ldots, \alpha_{\text{J}^{(*)},\Ks,\Ks},\\
\alpha_{\text{U}^{(*)},1,2}, \ldots, \alpha_{\text{U}^{(*)},1,\Ka},
\alpha_{\text{J}^{(*)},1,2}, \ldots, \alpha_{\text{J}^{(*)},1,\Ks}, \\
\alpha_{\text{U}^{(*)},2,3}, \ldots, \alpha_{\text{U}^{(*)},2,\Ka},
\alpha_{\text{J}^{(*)},2,3}, \ldots, \alpha_{\text{J}^{(*)},2,\Ka}, \ldots \\
\alpha_{\text{U}^{(*)},\Ka-1,\Ka}, \alpha_{\text{J}^{(*)},\Ks-1,\Ks})^\trasp \in [-1/2,1/2)^{L_\text{a}}
\end{multline}
whereas the latter one reads as
\be
\mathcal{L}_2:
\boldsymbol{\alpha}_\text{ord} \in [-1/2,1/2)^{L_\text{a}}
\longrightarrow \alphalabel \in [-1/2,1/2)^{L_\text{a}} \: .
\label{eq:L2}
\ee
Let us form the composite mapping
from $[-1/2,1/2)^{K}$ into $[-1/2,1/2)^{L_\text{a}}$, denoted
by $\mathcal{L}_1 \circ \mathcal{L}_2$:
estimation of $\boldsymbol{\nu}_\text{ord}$
from the observed vector
$\alphalabel$ is tantamount to inverting the
composite mapping $\mathcal{L}_1 \circ \mathcal{L}_2$.
Unfortunately, $\mathcal{L}_1 \circ \mathcal{L}_2$ is not invertible
unless some a priori knowledge is assumed.
Indeed, the linear transformation \eqref{eq:L1}
can be explicitly written as
\be
\boldsymbol{\alpha}_\text{ord}=\underbrace{\begin{pmatrix}
2 \, \bm I_{K} \\
\bm \Gamma
\end{pmatrix}}_{\bm B  \in \Rset^{L_\text{a} \times K}}  \boldsymbol{\nu}_\text{ord}
= \bm B \, \boldsymbol{\nu}_\text{ord}
\label{eq:sys-1}
\ee
where the matrix $\bm \Gamma \in \mathbb{R}^{(L_\text{a}-K) \times K}$
-- defined in \eqref{eq:Gamma} -- has exactly two entries 
in each row equal to $1$, and  all other entries equal to $0$.
By construction, the matrix $\bm{B}$ is full-column rank, i.e.,
$\rank(\bm{B})=K$.
On the other hand, the linear mapping \eqref{eq:L2} is represented by
the linear system
\be
\alphalabel = \Pblabel \, \boldsymbol{\alpha}_\text{ord}
\label{eq:sys-1}
\ee
where $\Pblabel \in \Rset^{L_\text{a} \times L_\text{a}}$ is an {\em unknown} permutation matrix,
which has been defined in the statement of the Proposition.
Since $\Pblabel$ is a permutation matrix, it results that
$\Pblabel^{-1}=\Pblabel^\trasp$ and, hence,
inversion of the composite mapping $\mathcal{L}_1 \circ \mathcal{L}_2$
is equivalent to solving the matrix equation
\be
\Pblabel \, \bm B \, \boldsymbol{\nu}_\text{ord} = \alphalabel
\label{eq:syst}
\ee
with respect to the unknowns $\Pblabel$ and $\boldsymbol{\nu}_\text{ord}$.

\section{Proof of Proposition~\ref{prop:2}}
\label{app:2}
The proof capitalizes on a
parameterization of the Toeplitz matrix
$\mathbf{T}_{\text{TX},k,b}$ -- defined in
 \eqref{eq:mat} with $b \in \{0,1\}$ -- in terms of {\em forward shift} $\Fb \in \Rset^{P \times P}$
and {\em backward shift} $\Bb \eqdef \Fb^\trasp \in \Rset^{P \times P}$ matrices \cite{Horn},
whose first column and the first row are given by
$[0, 1, 0, \ldots, 0]^\trasp$
and $[0, 1, 0, \ldots, 0]$, respectively.
One has
\barr
\mathbf{T}_{\text{TX},k,0} & = \sum_{\ell=0}^{\Lcp}
\psi(\ell \, T_\text{c}-\tau_{\text{TX},k}) \, \Fb^{\ell}
\label{eq:T0}
\\
\mathbf{T}_{\text{TX},k,1} & =
\sum_{\ell=1}^{\Lcp}
\psi(\ell \, T_\text{c}-\tau_{\text{TX},k}) \, \Bb^{P-\ell}
\label{eq:T1}
\earr
where $\tau_{\text{TX},k} \in [0, \Delta_\text{max}]$
and we remember that the pulse
$\psi(t) = \psi_\text{DAC}(t) \ast \psi_\text{ADC}(t)$ is known
at the CU.

It is crucial to note that
\be
\bm{\Phi}_{\text{TX},k} \eqdef \sum_{r=-1}^{1} \bm{\Phi}_{\text{TX},k}[r] =
g_{\text{TX},k}^2  \, \mathbf{C}_{\text{TX},k} \, \bm{\Omega} \, \bm{\Omega}^\trasp \,
\mathbf{C}_{\text{TX},k}^\trasp
\label{eq:Phi}
\ee
where $\mathbf{C}_{\text{TX},k} \eqdef  \mathbf{T}_{\text{TX},k,0}
+ \mathbf{T}_{\text{TX},k,1} \in \Rset^{P \times P}$ is a circulant matrix \cite{Horn} by construction,
whose first column is given by
\begin{multline}
\mathbf{c}_{\text{TX},k} \eqdef \Big(\psi(-\tau_{\text{TX},k}), \psi(T_\text{c}-\tau_{\text{TX},k}),
\ldots,
 \\
\psi(\Lcp \, T_\text{c}-\tau_{\text{TX},k}), 0, \ldots, 0 \Big)^\trasp \: .
\end{multline}
Using standard eigenstructure
concepts, one gets
\be
\mathbf{C}_{\text{TX},k}=\bm{W}_P \,  \diag(\vb_{\text{TX},k}) \,
\bm{W}_P^\herm
\ee
where the $p$th entry of
\be
\vb_{\text{TX},k} \eqdef \sqrt{P} \,\bm{W}_P^\herm \, \mathbf{c}_{\text{TX},k} \in \Cset^P
\ee
is given by $\{\vb_{\text{TX},k}\}_p= \Psi_{\text{TX},k}[p] \, e^{-j \frac{2 \pi}{T}    \tau_{\text{TX},k} \,p}$,  with  $\Psi_{\text{TX},k}[p]$ given by \eqref{eq:58}--\eqref{eq:59}
in Appendix~\ref{app:3}, for $p \in \{0,1,\ldots, P-1\}$.

In order to acquire the delay $\tau_{\text{TX},k} \in [0, \Delta_\text{max}]$,
for any $k \in \{1,2,\ldots, K_\text{TX}\}$ and $\text{TX} \in \{\text{U},\text{J}\}$,
we observe that, using the eigenstructure of  $\mathbf{C}_{\text{TX},k}$,
it results that
\be
\bm{W}_P^\herm \, \bm{\Phi}_{\text{TX},k} \, \bm{W}_P^*=
g_{\text{TX},k}^2 \, \Eb_{\text{TX},k} \, \bm{\Psi}_{\text{TX},k} \,
\bm{\Upsilon}  \, \bm{\Psi}_{\text{TX},k} \, \Eb_{\text{TX},k}
\ee
where we have defined the matrices
\barr
\Eb_{\text{TX},k} & \eqdef \diag\left(1, e^{-j \frac{2 \pi}{T} \tau_{\text{TX},k}}, \ldots,
e^{-j \frac{2 \pi}{T} \tau_{\text{TX},k} (P-1)} \right)
\\
\bm{\Psi}_{\text{TX},k} & \eqdef \diag\left(\Psi_{\text{TX},k}[0],
\Psi_{\text{TX},k}[1], \ldots, \Psi_{\text{TX},k}[P-1]\right)
\earr
and $\bm{\Upsilon}$ has been defined in \eqref{eq:ups}.

For $p \in  \{0,1,\ldots, P-1\}$, it is readily seen that
the $p$th diagonal entry of $\bm{W}_P^\herm \, \bm{\Phi}_{\text{TX},k} \, \bm{W}_P^*$
is given by
\be
\{\bm{W}_P^\herm \, \bm{\Phi}_{\text{TX},k} \, \bm{W}_P^*\}_{p,p}=
g_{\text{TX},k}^2 \, \Psi_{\text{TX},k}^2[p] \, \{\bm{\Upsilon}\}_{p,p}
\, e^{-j \frac{4 \pi}{T} \tau_{\text{TX},k} \,p} \: .
\ee
It is shown in Appendix~\ref{app:3} that, in the case of zero excess bandwidth,
$\Psi_{\text{TX},k}[0], \Psi_{\text{TX},k}[1], \ldots, \Psi_{\text{TX},k}[P/2-1]$ are equal to the corresponding
coefficients of the
DFT of $\psi(\ell \, T_\text{c})$ and, thus, they are known at the CU.
At this point, we are in the position of considering
the one-dimensional cost function defined in
\eqref{eq:I}. By resorting to the triangle inequality, one has
\be
\mathcal{I}_{\text{TX},k}(\beta) \le  |g_{\text{TX},k}|^2
\sum_{p=0}^{P/2-1} |\Psi_{\text{TX},k}[p]|^4 \, |\{\bm{\Upsilon}\}_{p,p}|^2
\ee
with equality if and
only if $\beta=\tau_{\text{TX},k}+i\, T/2$ ($i \in \Zset$).
This concludes the proof of the statement in Proposition~\ref{prop:2}.

\section{Expression of $\Psi_{\text{TX},k}[p]$}
\label{app:3}

Vector $\vb_{\text{TX},k}$ collects the coefficients of the $P$-point DFT of
$\mathbf{c}_{\text{TX},k}$. To evaluate such a DFT, we first observe that the
finite-length sequence $\{\psi(\ell \, T_\text{c}-\tau_{\text{TX},k})\}_{\ell=0}^{\Lcp}$
is the sampled version
of the waveform  $\psi(t-\tau_{\text{TX},k})$ at rate $1/T_\text{c}$.
The Fourier transform of $\psi(t-\tau_{\text{TX},k})$ is given by
$\Psi(f) \, e^{-j \, 2 \pi f \tau_{\text{TX},k}}$, where $\Psi(f)$
is the Fourier transform of $\psi(t)$ with two-sided
bandwidth $1/T_\text{c}$ approximately.\footnote{We assume zero excess bandwidth.
The derivations of this appendix can be straightforwardly generalized to the case
of a bandwidth greater than $1/T_\text{c}$.} By virtue of the Nyquist-Shannon sampling theorem, the
Fourier transform $\overline{\Psi}_{\text{TX},k}(\nu)$ of
the sequence $\psi(\ell \, T_\text{c}-\tau_{\text{TX},k})$ is given by
\be
\overline{\Psi}_{\text{TX},k}(\nu)= \frac{1}{T_\text{c}} \sum_{i=-\infty}^{+\infty}
\Psi\left(\frac{\nu-i}{T_\text{c}}\right) \, e^{-j \, 2 \pi \left(\frac{\nu-i}{T_\text{c}}\right) \tau_{\text{TX},k}} \: .
\ee
The DFT coefficients of
$\{\psi(\ell \, T_\text{c}-\tau_{\text{TX},k})\}_{\ell=0}^{\Lcp}$
are samples of $\overline{\Psi}_{\text{TX},k}(\nu)$
spaced in frequency at integer multiples of $1/P$. Thus, the $p$th entry of
$\vb_{\text{TX},k}$ reads as
\begin{multline}
\{\vb_{\text{TX},k}\}_p=\overline{\Psi}_{\text{TX},k}\left(\frac{p}{P}\right) = \left(e^{-j \frac{2 \pi}{T} \tau_{\text{TX},k}}\right)^p
\\ \cdot \underbrace{\frac{1}{T_\text{c}} \sum_{i=-\infty}^{+\infty}
\Psi\left(\frac{p-i \, P}{T}\right) \, e^{j \, \frac{2 \pi}{T_\text{c}} \chi_{\text{TX},k} i}}_{\Psi_{\text{TX},k}[p]}
\label{eq:58}
\end{multline}
for $p \in \{0,1,\ldots, P-1\}$.  Since $\Psi(f) \approx 0$ for $f \not \in (-0.5/T_\text{c},0.5/T_\text{c})$, it follows that\footnote{In practice, the integer $P=M+\Lcp$ is even since $M$ is a power of $2$ and $\Lcp=M/4$.}
\barr
\Psi_{\text{TX},k}[p] & \approx \frac{1}{T_\text{c}} \left[ \Psi\left(\frac{p}{T}\right) +
\Psi\left(\frac{p- P}{T}\right) \, e^{j \, \frac{2 \pi}{T_\text{c}} \chi_{\text{TX},k}}
\right]
\nonumber \\ &
= \begin{cases}
\frac{1}{T_\text{c}} \, \Psi\left(\frac{p}{T}\right) \:,
\\  \hspace{10mm} \text{for $p \in \left\{0, 1, \ldots, \frac{P}{2}-1\right\}$}
\\
\frac{1}{T_\text{c}} \, \Psi\left(\frac{p- P}{T}\right) \, e^{j \, \frac{2 \pi}{T_\text{c}} \chi_{\text{TX},k}} \:,
\\  \hspace{10mm} \text{for $p \in \left\{\frac{P}{2}-1, \frac{P}{2}, \ldots, P-1\right\}$} \: .
\end{cases}
\label{eq:59}
\earr
If $\tau_{\text{TX},k}$ is an integer
multiple of $T_\text{c}$, i.e., the fractional delay $\chi_{\text{TX},k}$ is zero, then
$\Psi_{\text{TX},k}[p]$ is the $p$th coefficient of the DFT of the sequence  $\psi(\ell \, T_\text{c})$.
On the other hand, when $\chi_{\text{TX},k} \neq 0$, the first $P/2$
coefficients $\Psi_{\text{TX},k}[0], \Psi_{\text{TX},k}[1], \ldots, \Psi_{\text{TX},k}[P/2-1]$ do not depend on the fractional delay and they are equal to the corresponding coefficients of the
DFT of $\psi(\ell \, T_\text{c})$,
whereas the remaining ones
$\Psi_{\text{TX},k}[P/2-1], \Psi_{\text{TX},k}[P/2], \ldots, \Psi_{\text{TX},k}[P-1]$
also depend on $\chi_{\text{TX},k}$.



\begin{IEEEbiography}
[{\includegraphics[width=1in,height=1.25in,clip,keepaspectratio]{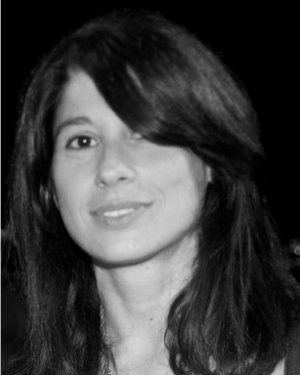}}]
{Donatella Darsena} (M'06-SM'16) received the Dr. Eng. degree summa cum laude in telecommunications engineering in 2001, and the Ph.D. degree in electronic and telecommunications engineering in 2005, both from the University of Napoli Federico II, Italy. From 2001 to 2002, she worked as embedded system designer in the Telecommunications, Peripherals and Automotive Group, STMicroelectronics, Milano, Italy. Since 2005, she has been an Assistant Professor with the Department of Engineering, University of Napoli Parthenope, Italy. Her research interests are in the broad area of signal processing for communications, with current emphasis on multicarrier modulation systems, space-time techniques for cooperative and cognitive communications, green communications for IoT. 
Dr. Darsena was an Associate Editor for the IEEE COMMUNICATIONS LETTERS from 
December 2016 to July 2019. 
She has served as Associate Editor for IEEE ACCESS since October 2018, 
Senior Area Editor for IEEE COMMUNICATIONS LETTERS 
since August 2019, and Associate Editor for 
IEEE SIGNAL PROCESSING LETTERS since 2020.

\end{IEEEbiography}

\vspace*{-2\baselineskip}

\begin{IEEEbiography}
[{\includegraphics[width=1in,height=1.25in,clip,keepaspectratio]{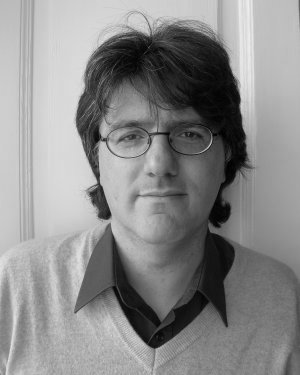}}]
{Giacinto Gelli}(M'18-SM'20)
was born in Napoli, Italy, on July 29, 1964.
He received the Dr. Eng. degree \textit{summa cum laude} in electronic
engineering in 1990, and the Ph.D. degree in computer science and
electronic engineering in 1994, both from the University of Napoli
Federico II.

From 1994 to 1998, he was an Assistant Professor with the
Department of Information Engineering, Second University of
Napoli.
Since 1998 he has been with the Department of Electrical Engineering and Information Technology, University of Napoli Federico II,
first as an Associate Professor,
and since November 2006 as a Full Professor of Telecommunications.
He also held teaching positions at the University Parthenope of
Napoli.
His research interests are in the broad area of
signal and array processing for communications,
with current emphasis on multicarrier modulation systems and
space-time techniques for cooperative and cognitive
communications systems.
 \end{IEEEbiography}

\vspace*{-2\baselineskip}

\begin{IEEEbiography}
[{\includegraphics[width=1in,height=1.25in,clip,keepaspectratio]{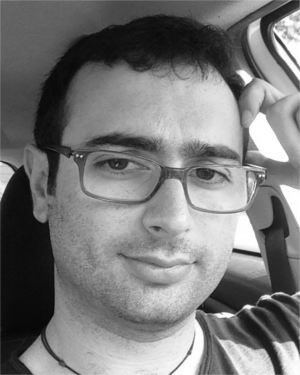}}]
{Ivan Iudice}
was born in Livorno, Italy, on November 23, 1986.
He received the B.S. and M.S. degrees
in telecommunications engineering in 2008 and 2010,
respectively, and the Ph.D. degree
in information technology and electrical engineering in 2017,
all from University of Napoli Federico II, Italy.

Since 2011, he has been with the Italian Aerospace Research Centre (CIRA), Capua, Italy. 
He first served as part of the Electronics and Communications laboratory and, since November 2020, 
he has served as part of the Security of Systems and Infrastructures laboratory.
His research activities lie in the area of
signal and array processing for communications,
with current interests are focused
on physical layer cyber security
and space-time techniques for
cooperative communications systems.
\end{IEEEbiography}

\vspace*{-2\baselineskip}

\begin{IEEEbiography}[
{\includegraphics[width=1in,height=1.25in,clip,keepaspectratio]{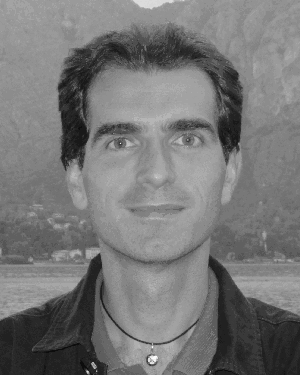}}]
{Francesco Verde}(M'10-SM'14) was born in Santa Maria Capua Vetere,
Italy, on June 12, 1974. He received the Dr. Eng. degree
\textit{summa cum laude} in electronic engineering
from the Second University of Napoli, Italy, in 1998, and the Ph.D.
degree in information engineering
from the University of Napoli Federico II, in 2002.
Since December 2002, he has been with the University of Napoli Federico II. He first served as an Assistant Professor of signal theory and mobile communications
and, since December 2011, he has served as an Associate Professor of telecommunications with the Department of Electrical Engineering and Information Technology.
His research activities include orthogonal/non-orthogonal multiple-access techniques, space-time processing for cooperative/cognitive communications, wireless systems optimization, and 
physical-layer security.

Prof. Verde has been involved in several technical program committees of major IEEE conferences in signal processing and wireless communications.
He has served as Associate Editor for IEEE TRANSACTIONS ON COMMUNICATIONS since 2017 and Senior Area Editor of the IEEE SIGNAL PROCESSING LETTERS since 2018. He was an Associate Editor of the IEEE TRANSACTIONS ON SIGNAL PROCESSING (from 2010 to 2014) and  IEEE SIGNAL PROCESSING LETTERS (from 2014 to 2018), as well as Guest Editor of the EURASIP Journal on Advances in Signal Processing in 2010 and SENSORS MDPI in 2018-2021. 
\end{IEEEbiography}

\end{document}